\documentclass[journal]{IEEEtran}
\usepackage{amsmath,amsfonts}
\usepackage{algorithmic}
\usepackage{algorithm}
\usepackage{array}
\usepackage[caption=false,font=normalsize,labelfont=sf,textfont=sf]{subfig}
\usepackage{textcomp}
\usepackage{stfloats}
\usepackage{url}
\usepackage{verbatim}
\usepackage{graphicx}
\usepackage{cite}

\usepackage[bookmarks=true]{hyperref}
\usepackage[capitalize, noabbrev]{cleveref}
\usepackage{booktabs}
\usepackage{multirow, amsthm}

\newtheorem{defi}{Definition}
\newtheorem{theorem}{Theorem}
\crefname{figure}{Fig.}{Figs.}
\crefname{table}{Table}{Tables}
\crefname{equation}{Equation}{Equations}

\begin{document}

\title{Transferring between Sparse and Dense Matching via Probabilistic Reweighting}

\author{Ya Fan and Rongling Lang\textsuperscript{$*$}
% Beihang University\\
% Beijing 100191, China\\
% {\tt\small fanya1502@buaa.edu.cn}
\thanks{*Corresponding author.}
\thanks{E-mail of Corresponding author: ronglinglang@buaa.edu.cn.}
\thanks{Ya Fan and Rongling Lang are with the School of Electronic Information Engineering, Beihang University, Beijing 100191, China.}
}

\markboth{IEEE TRANSACTIONS ON CIRCUITS AND SYSTEMS FOR VIDEO TECHNOLOGY}%
{Transferring between Sparse and Dense Matching via Probabilistic Reweighting}

\IEEEpubid{0000--0000/00\$00.00~\copyright~2021 IEEE}

\maketitle

\begin{abstract}
Detector-based and detector-free matchers are only applicable within their respective sparsity ranges.
To improve the adaptability of existing matchers, this paper introduces a novel probabilistic reweighting method.
Our method is applicable to Transformer-based matching networks and adapts them to different sparsity levels without altering network parameters. 
The reweighting approach adjusts attention weights and matching scores through feature detection probabilities.
We prove that the reweighted matching network is asymptotically equivalent to the detector-based matching network.
Furthermore, we propose a sparse training and pruning pipeline for detector-free networks based on reweighting.
Reweighted versions of SuperGlue, LightGlue, and LoFTR are implemented and evaluated on matching tasks including relative pose estimation and visual localization across different levels of sparsity.
Experiments show that the reweighted dense matching improves the pose accuracy of detector-based matchers in low-texture indoor scenes.
The performance of reweighted sparse LoFTR is comparable to detector-based matchers, demonstrating flexibility in balancing accuracy and computational complexity.

\end{abstract}    
\section{Introduction}

%##################################################################################################
Feature matching is a fundamental technology for many computer vision applications such as structure from motion (SFM)~\cite{schonbergerStructuremotionRevisited2016,gaoCompleteSceneReconstruction2020}, simultaneous localization and mapping (SLAM)~\cite{chenSurveyDeepLearning2020} and re-localization~\cite{lynenLargescaleRealtimeVisual2020}.
Recent advancements in learning-based feature matching have led to the emergence of two paradigms: detector-based matching~\cite{sarlinSuperGlueLearningFeature2020a,lindenbergerLightglueLocalFeature2023a} and detector-free matching~\cite{sunLoFTRDetectorFreeLocal2021,edstedtRoMaRobustDense2024}. 
For a pair of images to be matched, detector-based methods first utilize a detector~\cite{detoneSuperPointSelfSupervisedInterest2018a,tyszkiewiczDISKLearningLocal2020a} to extract two sets of feature points from the images, and then compute point-to-point correspondences within these sets. 
On the other hand, detector-free methods compute feature maps of the images and directly use all features for matching.

The two matching paradigms differ in the sparsity of features and exhibit complementarity. 
SuperPoint~\cite{detoneSuperPointSelfSupervisedInterest2018a} and SuperGlue~\cite{sarlinSuperGlueLearningFeature2020a} represent a classic detector-based matching pipeline, where SuperPoint extracts feature points, and SuperGlue aggregates these features to establish correspondences through optimal transport. 
Since SuperGlue typically takes sparse feature points as input, its matching process is consequently termed sparse matching.
Sparse matching is lightweight but prone to failures under significant variations in factors such as texture, viewpoint, illumination, and motion blur. 
With advancements in computational power, detector-free methods have gained increasing popularity. 
As a representative detector-free approach, LoFTR~\cite{wangEfficientLoFTRSemidense2024a} first extracts feature maps from images using a backbone network, then directly aggregates the features through a coarse-to-fine matching network. 
Since the coarse matching stage utilizes all features from the coarse feature maps, this stage is termed dense matching. 
Dense matching can leverage more information from the images, thereby being more resilient to variations that are difficult for sparse matching, albeit at a higher computational cost. 
In practical applications, where diverse image pairs pose varying degrees of matching difficulty, combining the two paradigms and dynamically adjusting the sparsity of matching according to the difficulty helps to balance precision and efficiency.
However, existing matching networks are only applicable within their respective sparsity ranges. 
Detector-based methods are exclusively trained on sparse feature points.
When the feature sparsity exceeds the sparsity range during training, the network encounters generalization issues, making it difficult to directly match the features. 
Among these methods, LightGlue~\cite{lindenbergerLightglueLocalFeature2023a} incorporates point pruning and depth pruning mechanisms to balance accuracy and efficiency, yet its feature points remain constrained within the sparse domain. 
Conversely, detector-free networks are fundamentally designed to operate within dense feature maps, and they lack effective mechanisms to enable sparse feature matching while preserving their original dense matching capabilities.

\IEEEpubidadjcol

It is noteworthy that the matching paradigms employ the attention mechanism~\cite{vaswaniAttentionAllYou2017}, which is inherently able to process an arbitrary number of features. 
This implies that the same matcher can modulate its matching sparsity simply by varying the quantity of input features.
However, due to the different distributions of features in sparse and dense matching, altering the number of features matched by the network causes generalization issues.
As illustrated in \cref{fig:transfer_to_dense} (a) and (b), SuperGlue struggles to generate a sufficient number of matches from sparse feature points for challenging image pairs. 
Expanding the feature points to dense representation and directly applying SuperGlue for matching these features would conversely lead to degraded matching performance or even failure to establish valid correspondences.

To adapt a pretrained attention network to a matching scenario where feature sparsity exceeds the training sparsity range, we propose a reweighting method in which the attention weights and matching scores are adjusted by feature detection probabilities. 
For detector-based networks, we first analyze the distribution shift and length generalization issues between sparse and dense matching from the perspective of random sequences. 
Subsequently, we propose to eliminate the distribution shift by matching random sequences whose marginal distributions are equal to the detection probability distribution.
We further prove that the reweighted matching constitutes the asymptotic limit of detector-based matching on random sequences. 
This theoretical result effectively transforms the matching of sufficiently long random sequences into a deterministic matching algorithm, where the sparsity of the matching is determined by the support of the probability distribution.
Based on our reweighting method, existing detector-based networks are transferred to a dense matching paradigm while accounting for the keypoint distribution estimated by the detector, as shown in~\cref{fig:transfer_to_dense} (c).
For detector-free networks, we propose a sparse training and pruning pipeline, optimizing the feature probabilities to allow them to perform sparse reweighted matching.
During training, a score head is trained to estimate sparse probability scores that maximize the reweighted matching performance. 
At inference time, pruning and reweighting are performed based on the probability scores, without altering the backbone or matching network parameters. 
This design preserves the original dense matching capability inherent in detector-free methods while simultaneously enabling efficient sparse matching functionality.
Moreover, our reweighted attention function is applicable to a wide range of Transformers. 
We reweight the mainstream SuperGlue, LightGlue, and LoFTR and evaluate them on matching tasks including relative pose estimation and visual localization at different sparsity levels.

\begin{figure}[tb]
  \centering
  \includegraphics[width=\linewidth, trim={0.0cm, 0.0cm, 0.0cm, 1.0cm}]{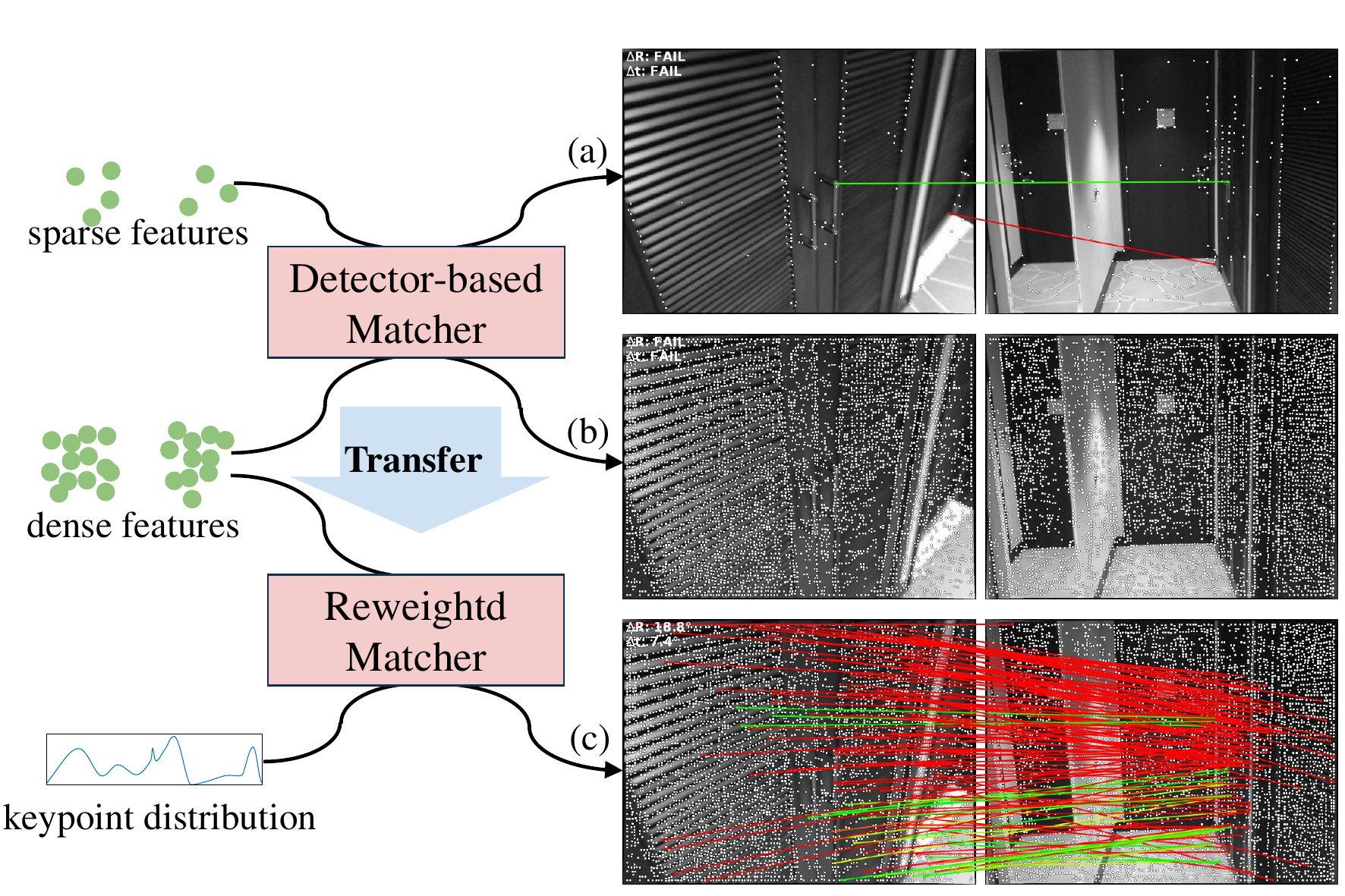}%
  \caption{Transferring detector-based matchers to dense matching.
  The original detector-based matcher struggles to match the challenging image pair with either sparse or dense feature points. 
  Given the keypoint distribution, our method adjusts attention weights and matching scores at inference time, and better matches dense feature points. 
  The results are from the original and reweighted SuperGlue~\cite{sarlinSuperGlueLearningFeature2020a} using SuperPoint~\cite{detoneSuperPointSelfSupervisedInterest2018a} features.
  }
  \label{fig:transfer_to_dense}%
\end{figure}

The main contributions of our work are as follows:
\begin{itemize}
	\item We propose a novel probabilistic reweighting method that enables sparsity transfer for Transformer-based matching networks without parameter modification, and theoretically prove that it constitutes the asymptotic limit of detector-based matching on long random sequences;
	\item By implementing reweighted detector-based matchers including SuperGlue and LightGlue, we demonstrate superior indoor performance of reweighted dense matching over original sparse paradigms in relative pose estimation and visual localization tasks;
	\item We present a reweighting-based sparse training and pruning pipeline for detector-free matchers, which enables sparse matching and preserves their dense performance, achieving comparable accuracy to sparse baselines while offering efficiency-accuracy trade-offs complementary to LightGlue.
        % We propose a novel probabilistic reweighting method applicable to Transformer-based matching networks, enabling transfer across different sparsity levels without modifying network parameters. We theoretically prove that the reweighted matching network constitutes the asymptotic limit of the detector-based matching network on sufficiently long random sequences;
        % We present a reweighting-based sparse training and pruning pipeline that enables sparse matching capability for detector-free matchers while maintaining their original dense performance. Relative pose estimation experiments show that the reweighted sparse LoFTR achieves comparable performance to existing sparse matchers, while exhibiting complementary accuracy-efficiency trade-offs with LightGlue.
        % We implement reweighted versions of detector-based matchers including SuperGlue and LightGlue. Experimental results on relative pose estimation and visual localization demonstrate that our reweighted dense matching paradigm outperforms the original sparse matching in indoor scenarios;
\end{itemize}
\section{Related Work}\label{sec:related}
\subsection{Detector-based Matching} 
Classical detector-based matching methods focus on designing keypoints with correspondences obtained through robust matching~\cite{fischlerSurveyDeepLearningsample1981,liuBBHomographyJointBinary2015}. 
FAST~\cite{rostenMachineLearningHighSpeed2006} and Harris~\cite{harrisCombinedCornerEdge1988} pioneered keypoint detection.
SIFT~\cite{loweDistinctiveImageFeatures2004} employed a scale-space Difference-of-Gaussian operator to localize scale-invariant keypoints.
SURF~\cite{baySURFSpeededRobust2006} accelerated feature detection through integral images and Hessian matrix approximation.
Two corner detection methods were presented based on Log-Gabor wavelet transform~\cite{gaoMultiscaleCornerDetection2007a}.
BRIEF~\cite{calonderBRIEFComputingLocal2012} used random sampled point pairs for binary description.
Subsequent research~\cite{rubleeORBEfficientAlternative2011, liuFeaturesCombinedBinary2020} efforts further enhanced handcrafted binary descriptors.
% And 
% Most of these works treat feature point detection as a deterministic process, with limited probabilistic interpretation of saliency maps.

Recent approaches train neural networks for keypoint detection and description. 
TILDE~\cite{verdieTILDETemporallyInvariant2015} trained a regressor to predict a score map with SIFT ground truth.
SuperPoint~\cite{detoneSuperPointSelfSupervisedInterest2018a} first established synthetic dataset to guide detector training, then bootstrapped the score map on real images with homographic adaption strategy, and its descriptors were trained with triplet loss.
D2-Net~\cite{dusmanuD2NetTrainableCNN2019} and R2D2~\cite{revaudR2d2ReliableRepeatable2019} introduced joint learning for detection and description pipelines.
% interprets the score map of keypoint detectors from a probabilistic perspective. 
DISK~\cite{tyszkiewiczDISKLearningLocal2020a} and PosFeat~\cite{liDecouplingMakesWeakly2022} employed probabilistic method to relax cycle-consistent matching and simplified reinforce learning.
MicKey~\cite{barroso-lagunaMatching2DImages2024a} proposed a random sampling strategy with differentiable RANSAC~\cite{fischlerRandomSampleConsensus1981} to optimizes 3D keypoint coordinates.
SketchDesc~\cite{yuSketchDescLearningLocal2021} introduced learned descriptor to match sketches.
TCDesc~\cite{panTCDescLearningTopology2022} improved topology consistency of descriptors.
AANet~\cite{raoLearningGeneralDescriptors2023} adopted attentional aggregation to boost the informativeness of descriptors.
DomainFeat~\cite{xuDomainFeatLearningLocal2024} adapted feature points to various image domains.
For efficient feature points, CDbin~\cite{yeCDbinCompactDiscriminative2020} designed shallow network for binary description.
And XFeat~\cite{potjeXFeatAcceleratedFeatures2024} introduced lightweight detection and description network that enables efficient semi-dense matching.

Learning-based feature matching methods have been proposed to better aggregate information of feature points. 
SuperGlue~\cite{sarlinSuperGlueLearningFeature2020a} used a Graph Neural Network (GNN) for feature matching.
The GNN of SuperGlue is specifically a Transformer~\cite{vaswaniAttentionAllYou2017} which is effective on both handcrafted and learned feature points. 
Scale-Net~\cite{fuLearningReduceScale2023a} focused on covisible image areas for matching with large scale changes.
Recent works have sought to enhance the efficiency of feature matching.
Seeded Graph Matching Network~\cite{chenLearningMatchFeatures2021} leveraged a seeding mechanism to sparsify fully connected graphs.
ClusterGNN~\cite{shiClustergnnClusterbasedCoarsefine2022} adopted a progressive clustering strategy to partition feature points into subgraphs, reducing quadratic complexity.
AMatFormer~\cite{jiangAMatFormerEfficientFeature2024} reduced computation cost by selecting anchors.
Notably, LightGlue~\cite{lindenbergerLightglueLocalFeature2023a} enables adaptive computation and demonstrates gains in both efficiency and accuracy.
Although detector-based matching approaches are computationally efficient, they often suffer from performance degradation in low-texture areas owing to the insufficient density of feature points.

\subsection{Detector-Free Matching} 
Unlike detector-based matching, detector-free matching employs dense descriptors or features of images for matching. 
SIFT Flow~\cite{liuSiftFlowDense2010} achieved scene alignment using densified classical descriptors. 
Early learning-based methods utilized contrastive loss to learn dense descriptors~\cite{choyUniversalCorrespondenceNetwork2016a,schmidtSelfSupervisedVisualDescriptor2017}. 
More recent approaches learned dense matching using dense 4D cost volumes~\cite{roccoNCNetNeighbourhoodConsensus2022,liDualresolutionCorrespondenceNetworks2020a,truongLearningAccurateDense2021}. 

Transformers are also utilized to match dense features.
LoFTR~\cite{sunLoFTRDetectorFreeLocal2021} leveraged self and cross linear attention to aggregate feature descriptors conditioned on two-view images and achieved promising performance.
COTR~\cite{jiangCOTRCorrespondenceTransformer2021} matched images jointly by self attention.
Convolutional backbone network is required in both LoFTR and COTR to extract dense features.
Matchformer~\cite{wangMatchformerInterleavingAttention2022} merged convolution into Transformer as positional patch embedding, and proposed an extract-and-match scheme.
Wang et al. adopted convolution and max-pooling into attention to aggregate features~\cite{wangEfficientLoFTRSemidense2024a}. 
ASpanFormer~\cite{chenASpanFormerDetectorFreeImage2022} incorporated flow and uncertainty estimation to guide adaptive attention.
QuadTree~\cite{tangQuadTreeAttentionVision2022} restricts the attention span during hierarchical attention.
For dense geometric matching, PMatch~\cite{zhuPMatchPairedMasked2023} leveraged paired masked image modeling with extended LoFTR module, and PATS~\cite{niPatsPatchArea2023} used patch area transportation and subdivision to tackle large scale differences.
Probability functions are adopted to model the dense match.
PDC-Net++~\cite{truongPDCNetEnhancedProbabilistic2023} employed mixture Laplacian distribution and DKM~\cite{edstedtDKMDenseKernelized2023} used Gaussian processes for regression of coarse match coordinates.
RoMa~\cite{edstedtRoMaRobustDense2024} replaced the backbone and coordinate decoder of DKM with a Transformer and demonstrated its robustness.
While detector-free methods achieve strong performance, they are computationally intensive, and the scheme that matches all image features exhibits limited flexibility.

\subsection{Length Generalization of Transformers} 
In our framework, features with varying levels of sparsity or quantity are treated as sequences of different lengths. 
And Transformers are tasked with matching features whose sparsity deviates from the training distribution, which inherently embodies its length generalization capability.
In the field of natural language processing (NLP), the length generalization of Transformers primarily refers to length extrapolation~\cite{pressTrainShortTest}. 
Developing better position encoding is an effective approach to enhancing the length extrapolation~\cite{neishiRelationPositionInformation2019, ruossRandomizedPositionalEncodings2023}. 
It has been found that sinusoidal absolute positional encoding struggles with extrapolation~\cite{daiTransformerXLAttentiveLanguage2019}, while relative positional encoding exhibits greater robustness to changes in input length\cite{likhomanenkoCapeEncodingRelative2021, chiKerpleKernelizedRelative2022}. 
RoPE~\cite{suRoformerEnhancedTransformer2024} was proposed to multiply keys and queries by rotation matrices. 
Despite the absolute nature of this rotary process, the attention mechanism of RoPE depends solely on relative distance, which benefits length extrapolation. 
In specific tasks, Transformers have been proven to achieve length generalization~\cite{xiaoTheoryLengthGeneralization2024}. 
Anil et al.~\cite{anilExploringLengthGeneralization2022} found that fine-tuning regimes, scaling data, model sizes, and computational resources do not improve length generalization, whereas methods like chain-of-thought~\cite{weiChainofthoughtPromptingElicits2022} in the in-context learning regime do. 
LightGlue~\cite{lindenbergerLightglueLocalFeature2023a} adopted RoPE from NLP, but its extrapolation performance on dense image features~\cite{gaoSARTargetIncremental2024,kongFewShotClassIncrementalSAR2025} remains underexplored.
We further explore length generalization in the domain of image matching for sparsity transfer.
\section{Methodology}
We focus on image matchers using Transformers and investigate the feature matching process from a sequential perspective.
We then derive probabilistic reweighted attention and matching functions from a detector-based perspective and prove their asymptotic equivalence to the original functions.
The probabilistic derivation of the reweighted Transformer is illustrated in \cref{fig:overview}.
Based on the reweighting method, we transfer existing matchers between sparse and dense paradigms.

\subsection{Preliminaries}\label{sec:preliminaries}
Both sparse and dense image matching use image features, where a backbone is employed to compute feature map $\mathbf{D}$. 
In sparse matching, feature points are detected and sampled from the feature map before matching. 
On the other hand, dense matching directly computes correspondences using the entire feature map.

Given an image $I$, we regard the keypoint detection process as a procedure of independent and identically distributed (i.i.d.) sampling.  
A set of keypoints can be represented as a random sequence $S_{I}$ composed of random pixel coordinates $x_1, x_2, ..., x_n$, i.e., $S_{I} = \{x_n\}$, where $n$ is the number of keypoints. 
Each coordinate $x$ is a random variable distributed across the entire image, with a marginal probability distribution denoted as $P(x|I, \theta_F)$, where $\theta_F$ represents the parameters of the feature detector.
For the feature map $\mathbf{D}$ of image $I$, the feature or descriptor at $x$ is denoted as $\mathbf{D}(x)$. 
The sequence $S_{I}$ thus yields a sequence of feature points $F_{I}=\{(x_1, \mathbf{D}(x_1)), (x_2, \mathbf{D}(x_2)),...,(x_n, \mathbf{D}(x_n))\}$. 
We denote $F_I^*$ as a sequence of all possible feature points form $\mathbf{D}$.
Let the $i$-th point in $F_I$ be $F_{I,i}$.
For each index $i$, there is a corresponding index $i^*$ in $F_I^*$ such that $F_{I,i} = F_{I,i^*}^*$.

Notably, i.i.d. sampling may generate duplicate feature points. 
For short sequences where the repetition probability remains low, this sampling behavior approximates the non-repetitive sampling strategies typically adopted in keypoint detection.
As sequence length extends sufficiently, the repetition frequency of features naturally reflects their probability importance.
Our reweighting method exploits this importance measure for sparsity transfer.

For images $A$ and $B$, sparse matching establishes correspondences between a limited number of feature points $F_A$ and $F_B$, while dense matching involves matching $F_A^*$ and $F_B^*$ without probabilistic sampling.
In both sparse and dense matching, the features are first aggregated utilizing Transformers, followed by the computation of assignment matrices through matching functions.
Due to differences in sequence lengths, sparse and dense matching Transformers are typically isolated, with distinct sets of parameters.
In this paper, we propose deriving sparse and dense matchers without changing Transformer parameters, based on a reweighting approach.
We denote $\mathrm{TF_{sparse}}$ and $\mathrm{TF_{dense}}$ as sparse and dense matching mode of the Transformer, respectively, and the aggregated feature points are expressed as 
\begin{equation}\label{eqn:feature_update}
  \begin{split}
    &  (F_A', F_B') = \mathrm{TF_{sparse}}(F_A, F_B) \\
    &  (F_A^{*'}, F_B^{*'}) = \mathrm{TF_{dense}}(F_A^*, F_B^*),
  \end{split}
\end{equation}

where the order of feature points in each sequence is maintained. 
For each index $i$, $F_{A,i}'$ is the update of $F_{A,i}$, whose descriptor is replaced by corresponding token of the network's last layer. 
The same applies to $F_B$, $F_A^*$ and $F_B^*$.

For matching Transformers trained on short sequences $F_A, F_B$, directly applying them to $F_A^*, F_B^*$ may lead to generalization challenges. 
As an interpretation based on random sequences, we will demonstrate in \cref{sec:rw_attention} that matching $F_A^*, F_B^*$ is equivalent to matching sufficiently long random sequences $F_A^{U}$ and $F_B^{U}$, where both sequences are sampled i.i.d. from a uniform distribution. 
Two key differences exist between $F_A, F_B$ and $F_A^{U}, F_B^{U}$: (1) the marginal distributions mismatch, and (2) the sequence lengths are substantially different. 
These differences introduce both distributional shift and length generalization issues.
We focus on eliminating the marginal distribution differences by matching $F_A^{P}$ and $F_B^{P}$, where $F_A^{P}$ and $F_B^{P}$ are sufficiently long sequences sampled i.i.d. from $P(x|A,\theta_F)$ and $P(x|B,\theta_F)$, respectively. 
The sequences $F_A, F_B$ and $F_A^{P}, F_B^{P}$ share identical marginal distributions, differing only in sequence length.
And we hypothesize that matching $F_A^{P}$ and $F_B^{P}$ can achieve better generalization than matching $F_A^{U}$ and $F_B^{U}$. 
In \cref{sec:rw_attention,sec:rw_matching_function}, we propose an equivalent reweighting algorithm that is deterministic and practically executable for matching $F_A^{P}$ and $F_B^{P}$.

\begin{figure*}[!t]
  \centering
  \includegraphics[width=0.85\linewidth]{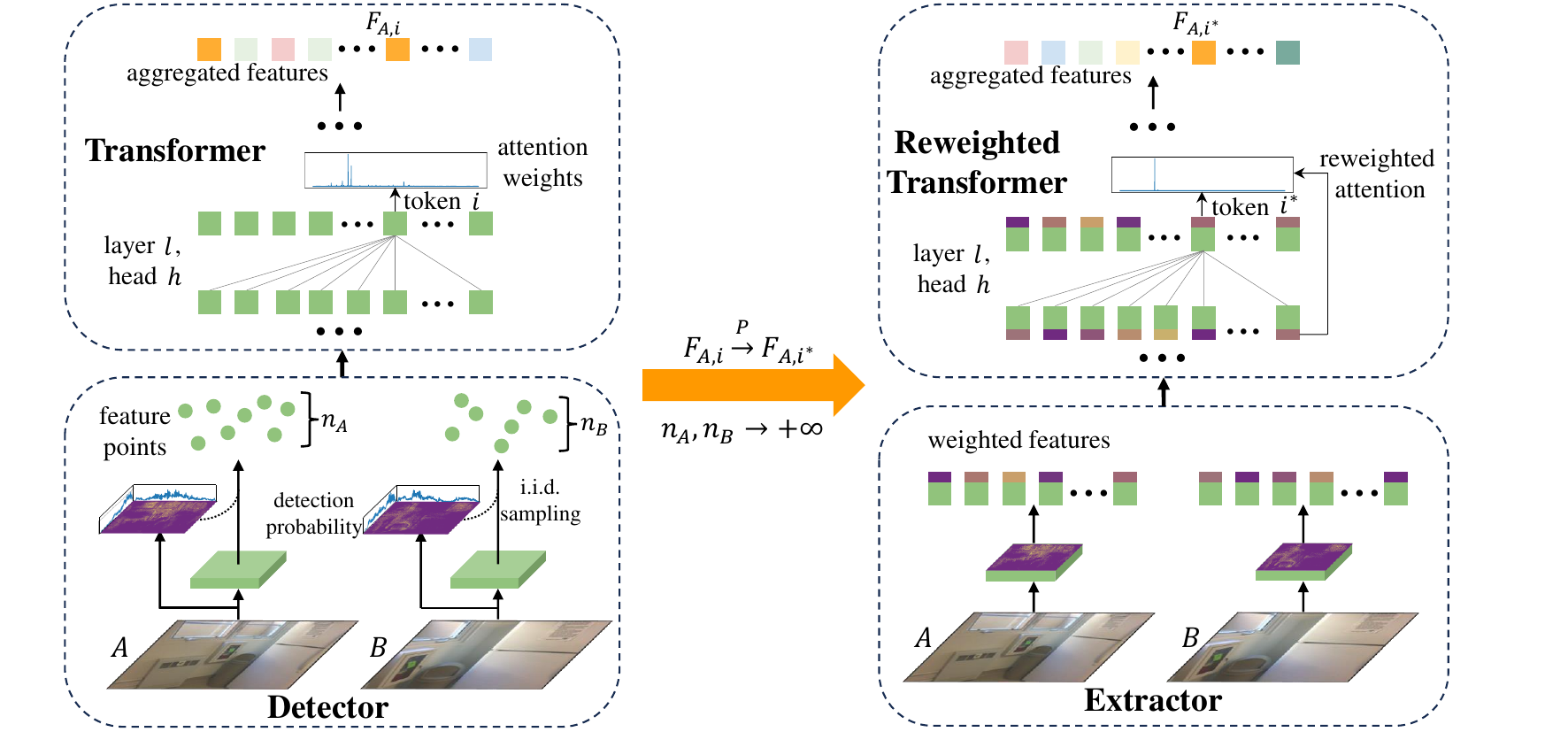}%
  \caption{Probabilistic derivation of reweighted Transformer.
  The left panel depicts the conventional detector-based feature matching, where the detector extracts features and computes detection probabilities. 
  A Transformer is utilized to aggregate the keypoints sampled by the detector. 
  The right panel illustrates the reweighted attention network that employs features with probability weights and adjusts attention weights based on the probabilities.
  The reweighted matching network constitutes the asymptotic limit of conventional matching network: for i.i.d. keypoints, as the sample size approaches infinity, the features aggregated by the original network converge in probability to those aggregated by the reweighted network.
  }
  \label{fig:overview}
\end{figure*}

\subsection{Reweighted Attention}\label{sec:rw_attention}
Matching Transformers aggregate features using an embedding layer and multiple attention layers.
Given feature points $F_A$ and $F_B$ with cardinalities of $n_A$ and $n_B$, the embedding layer maps feature points to tokens, $\mathbf{X}_A \in \mathbb{R}^{d \times n_A}$, $\mathbf{X}_B \in \mathbb{R}^{d \times n_B}$, where $d$ is the number of feature channels.
For each attention layer, the input tokens are denoted as $\mathbf{X}_Q$, $\mathbf{X}_K$ and $\mathbf{X}_V$. 
The output token $\mathbf{X}_Q'$ is computed as follows:
\begin{equation}\label{eqn:attention}
  \begin{split}
  &  \mathbf{X}'= \mathbf{X}_Q + \sum_{h = 1}^{H} \mathbf{W}_{VO}^{h} \mathbf{X}_V \mathrm{Sim}(\mathbf{W}_K^{h}\mathbf{X}_K, \mathbf{W}_Q^{h}\mathbf{X}_Q) \\
  & \mathbf{X}_Q' = \mathbf{X}' + \mathbf{W}_2 \mathrm{Act} (\mathbf{W}_1 \mathbf{X}' + \mathbf{b}_1 \mathbf{1}^\top) + \mathbf{b}_2 \mathbf{1}^\top
  \end{split}
\end{equation}
where $H$ is the number of attention heads, and $\mathbf{W}_{VO}^{h}$, $\mathbf{W}_K^{h}$, and $\mathbf{W}_Q^{h}$ are parameter matrices of the $h$-th attention head. 
In feed forward part, $\mathbf{W}_1, \mathbf{W}_2$ and $\mathbf{b}_1, \mathbf{b}_2$ are parameter matrices and vectors respectively, and all ones vector $\mathbf{1}$ broadcasts parameter vectors to match the shape of tokens.
The attention matrix is obtained by calculating the query and key matrices with the function $\mathrm{Sim}$.
The activation function $\mathrm{Act}$ can be any continuous element-wise function.  
For brevity, we omitted layer normalization (LayerNorm). 
Nevertheless, our subsequent theoretical analysis remains applicable to cases involving LayerNorm, as token-wise functions do not affect the derivation.

\begin{defi}
Let $\delta$ be a similarity function that maps $\mathbb{R}^d \times \mathbb{R}^d$ into $\mathbb{R}$. 
Given key and query matrices $\mathbf{K} \in \mathbb{R}^{d \times n_K} , \mathbf{Q} \in \mathbb{R}^{d \times n_Q}$, denote their column vectors as $\mathbf{K}_i$ and $\mathbf{Q}_j$.
Then the attention matrix $\mathrm{Sim}(\mathbf{K},\mathbf{Q}) \in \mathbb{R}^{n_K \times n_Q}$ is defined such that
\begin{equation}
    \mathrm{Sim}_{i,j}(\mathbf{K},\mathbf{Q}) = \frac{{\delta({\mathbf{K}_i},{\mathbf{Q}_j})}}{{\sum\limits_{k = 1}^{{n_K}} {\delta({\mathbf{K}_k},{\mathbf{Q}_j})} }}, \\
\end{equation}
where $\mathrm{Sim}_{i,j}$ is the element in its $i$-th row and $j$-th column, $1 \le i \le {n_K},1 \le j \le {n_Q}$.
\end{defi}

In standard attention where Softmax is employed, the similarity function is
\begin{equation}
  \delta_{\mathrm{softmax}}({\mathbf{K}_i},{\mathbf{Q}_j}) = \exp (\tau^{-1}{\mathbf{K}_i}^\top{\mathbf{Q}_j}),
\end{equation}
where $\tau$ is temperature parameter. In linear attention, the similarity function is 
\begin{equation}
  \delta_{\mathrm{linear}}({\mathbf{K}_i},{\mathbf{Q}_j}) = \phi({\mathbf{K}_i})^\top \phi({\mathbf{Q}_j}),
\end{equation}
where $\phi$ is a continuous element-wise function. 

\begin{defi}
Given key and query matrices $\mathbf{K}$ and $\mathbf{Q}$, $\mathbf{K} \in \mathbb{R}^{d \times n_K}, \mathbf{Q} \in \mathbb{R}^{d \times n_Q}$. 
The probability weights associated with each column of the key matrix are denoted as $p_I(i)$, $p_I(i) = P(x_i|I,\theta_F), 1 \leq i \leq n_K$, where $I$ is the image corresponding to the key matrix.
Then the reweighted attention matrix $\mathrm{Sim}^p(\mathbf{K},\mathbf{Q}) \in \mathbb{R}^{n_K \times n_Q}$ is defined such that
\begin{equation}
  \mathrm{Sim}_{i,j}^{p_I}(\mathbf{K},\mathbf{Q}) = \frac{p_I(i){\delta({\mathbf{K}_i},{\mathbf{Q}_j})}}{{\sum\limits_{k = 1}^{{n_K}} p_I(k){\delta({\mathbf{K}_k},{\mathbf{Q}_j})} }}. 
  \label{eq:reweighted_attention}
\end{equation}
\end{defi}

This formula can be intuitively interpreted as assigning higher attention scores to tokens with greater probability weights in the reweighted attention mechanism, which helps to retain the knowledge about probabilistic importance of the features.
The following theorem demonstrates that the reweighted attention layer is the asymptotic limit of original attention layer.

\begin{theorem} \label{thm:reweighted_attention_layer}
  Given tokens $\mathbf{Q}^*$, $\mathbf{K}^*$ and $\mathbf{V}^*$, and an attention layer described in \cref{eqn:attention},
  $\mathbf{Q}^* \in \mathbb{R}^{d \times n^*}, \mathbf{K}^* \in \mathbb{R}^{d \times m^*}, \mathbf{V}^* \in \mathbb{R}^{d \times m^*}$.
  Denote $ \{ 1,\ldots , n \}$ as $[n]$.
  For $m,n \in \mathbb{N}^+ $, which approach infinity, suppose there are sequences of random tokens $\mathbf{Q}^{(n)}$, $\mathbf{K}^{(m)}$ and $\mathbf{V}^{(m)}$ that satisfy:

  (a) The sequences extend with $m,n$: 
  \begin{equation}
      \label{eqn:condition_shape}
       \mathbf{Q}^{(n)} \in \mathbb{R}^{d \times n}, \mathbf{K}^{(m)}, \mathbf{V}^{(m)} \in \mathbb{R}^{d \times m}.
  \end{equation}

  (b) The $j$-th token of $\mathbf{Q}^{(n)}$ converges in probability to $\mathbf{Q}_{j^*}^*$:
  \begin{equation}
      \label{eqn:condition_limQ}
      \forall j \in \mathbb{N}^+, \exists j^* \in [n^*]: \mathbf{Q}_{j}^{(n)} \xrightarrow{P} \mathbf{Q}_{j^*}^*.
  \end{equation}

  (c) The $i$-th token of $\mathbf{K}^{(m)}$ and the $i$-th token of $\mathbf{V}^{(m)}$ converge in probability to $\mathbf{K}_{i^*}^*$ and $ \mathbf{V}_{i^*}^*$, respectively:
  \begin{equation}
      \label{eqn:condition_limKV}
       \begin{split}
          &  \forall i \in \mathbb{N}^+, \exists i^* \in [m^*]: \\
          & \mathbf{K}_{i}^{(m)} \xrightarrow{P} \mathbf{K}_{i^*}^*, \mathbf{V}_{i}^{(m)} \xrightarrow{P} \mathbf{V}_{i^*}^*.
       \end{split}
  \end{equation}

  (d) Tokens that converge in probability to the same token are equal, and their proportion in the sequence converges in probability. Formally, for index sets $\mathcal{I}_m(i^*) = \{k | k \in [m], k^* = i^* \}$, the tokens satisfy:
  \begin{equation}
      \label{eqn:ratio_converges_to_probability}
      \begin{split}
      & \forall i^* \in [m^*], \forall i, i' \in \mathcal{I}_m(i^*): \\
      & \mathbf{K}_{i}^{(m)} = \mathbf{K}_{i'}^{(m)}, \mathbf{V}_{i}^{(m)} = \mathbf{V}_{i'}^{(m)}, \\
      & \frac{1}{m} \left\lvert \mathcal{I}_m(i^*) \right\rvert \xrightarrow{P} p(i^*), 
      \end{split}
  \end{equation}
  where $p(i^*)$ is the probability associated with $\mathbf{K}_{i^*}^*$ and $\mathbf{V}_{i^*}^*$.

  Let $\mathbf{Q}^{'(m,n)}$ be the updated token of the attention layer, with $\mathbf{Q}^{(n)}$, $\mathbf{K}^{(m)}$ and $\mathbf{V}^{(m)}$ as input.
  And let $\mathbf{Q}^{'*}$ represent the updated token of the reweighted attention layer, which replaces $\mathrm{Sim}$ with $\mathrm{Sim}^p$ in \cref{eq:reweighted_attention}, and takes $\mathbf{Q}^*$, $\mathbf{K}^*$ and $\mathbf{V}^*$ as input.
  Then for any $j \in \mathbb{N}^+ $, as $m, n$ approach infinity, we have
  \begin{equation}
      \mathbf{Q}_{j}^{'(m,n)} \stackrel{P}{\longrightarrow} \mathbf{Q}^{'*}_{j^*}.
  \end{equation}

\end{theorem}
Condition (d) of the theorem essentially represents the law of large numbers, and the i.i.d. sampling process of the feature points extraction satisfies this condition.

For image $A$, $B$ and probability weights $p_A$, $p_B$, the reweighted Transformer $\mathrm{TF}^{p_A, p_B}$ of network $\mathrm{TF}$ is defined as follows: while keeping the parameters unchanged, the attention function of each attention layer is replaced with the reweighted version. 
In each reweighted attention, the weights $p_A$ or $p_B$ are used based on whether the token $\mathbf{X}_K$ corresponds to the image $A$ or $B$, respectively. 
We introduce the following theorem that reveals the asymptotic equivalence of the reweighted network that matches $F_A^*, F_B^*$ and the original network that matches $F_A^{P}, F_B^{P}$:

\begin{theorem}\label{thm:rw_attention}
Let $F_A$ and $F_B$ be two sequences of i.i.d. feature points, which are sampled from $F_A^*$ and $F_B^*$ according to detection probabilities $p_A$ and $p_B$, respectively.
Let $\mathrm{TF}$ and $\mathrm{TF}^{p_A, p_B}$ be an attention network and its reweighted version, as described above. 
The feature points are aggregated as described in \cref{eqn:feature_update}, where $\mathrm{TF_{sparse}} = \mathrm{TF}$ and $\mathrm{TF_{dense}} = \mathrm{TF}^{p_A, p_B}$.
For any index pair $i$ and $j$, denote $i$-th and $j$-th output feature points from network $\mathrm{TF}$ as $\mathrm{TF}_{i,j}$, i.e., $\mathrm{TF}_{i,j}(F_A, F_B) = (F'_{A,i}, F'_{B,j})$.
As the lengths of $F_A$ and $F_B$ tend to infinity, we have
% \begin{equation}
%   \lim_{n_A, n_B \to \infty} \mathrm{TF}_{i,j}(F_A, F_B) \stackrel{P}{=} \mathrm{TF}_{i^*,j^*}^{p_A, p_B}(F_A^*, F_B^*).
% \end{equation}
\begin{equation}
 \mathrm{TF}_{i,j}(F_A, F_B) \stackrel{P}{\longrightarrow} \mathrm{TF}_{i^*,j^*}^{p_A, p_B}(F_A^*, F_B^*).
\end{equation}
\end{theorem}
The theorem states that as the number of sampled feature points approaches infinity, the outputs from directly aggregation of these feature points via an attention network converge in probability to those from processing the deduplicated features with the reweighted network. 
Therefore, our reweighting method transforms the problem of matching infinite random sequences into a practically feasible deterministic matching.
When $p_A$ and $p_B$ are uniform distributions, we have $\mathrm{TF}^{p_A, p_B} = \mathrm{TF}$, which demonstrates that matching $F_A^*$ and $F_B^*$ is equivalent to matching $F_A^{U}$ and $F_B^U$ in~\cref{sec:preliminaries}.

\subsection{Reweighted Matching Function}\label{sec:rw_matching_function}
A matching function computes an assignment matrix $\mathbf{P}$ of all possible matches.
Frequently employed matching layers leverage features from multi-layer attention to compute $\mathbf{P}$, whereas in reweighted matching layers, the functions incorporate the probabilities of these features as well.
We introduce two types of reweighted matching functions with optimal transport and dual-softmax.
Then we demonstrate their asymptotic equivalence to original functions through a probabilistic convergence theorem.

For aggregated feature points $F_A^{'}$ and $F_B^{'}$, denote the $i$-th and $j$-th feature vectors as $\mathbf{f}_{A,i}$ and $\mathbf{f}_{B,j}$, then a score matrix $\mathbf{S}$ is defined as pairwise inner product of the feature vectors:
\begin{equation}
  \mathbf{S}_{i, j} = \mathbf{f}_{A,i}^\top \mathbf{f}_{B,j}.
\end{equation}
For matching function with optimal transport, the score matrix $\mathbf{S}$ is augmented to $\overline{\mathbf{S}}$ by appending a new row and column filled with a dustbin score, as in~\cite{sarlinSuperGlueLearningFeature2020a}.
The probabilistic weight vectors of $A$ and $B$ are denoted as $\mathbf{a}$ and $\mathbf{b}$, $\mathbf{a} \in \mathbb{R}^{n_A + 1}, \mathbf{b} \in \mathbb{R}^{n_B + 1}$ where the last component of each vector represents the dustbin:
\begin{equation}
  \begin{split}
   & \mathbf{a}_i = p_A(i), \mathbf{a}_{n_A + 1} = 1, 1 \leq i \leq n_A \\
   & \mathbf{b}_j = p_B(j), \mathbf{b}_{n_B + 1} = 1, 1 \leq j \leq n_B.    
  \end{split}
\end{equation}
The reweighted optimal transport layer finds an augmented assignment $\overline{\mathbf{P}}$ to maximize the total score $\sum_{i,j} \overline{\mathbf{S}}_{i,j} \overline{\mathbf{P}}_{i,j}$ with the following constraints:
\begin{equation}\label{eqn:ot_matching}
  \overline{\mathbf{P}} \mathbf{1}_{n_B+1} = \mathbf{a} \quad and \quad \overline{\mathbf{P}}^\top \mathbf{1}_{n_A+1} = \mathbf{b}.
\end{equation}
In practice, the optimal transport layer utilizes the Sinkhorn algorithm~\cite{sinkhornConcerningNonnegativeMatrices1967,cuturiSinkhornDistancesLightspeed2013} to compute soft assignments. 
We define $\mathrm{OT}^{p_A, p_B}(F_A', F_B')$ as the assignment computed by the layer, where $F_A'$ and $ F_B'$ are the features and $ p_A$ and $ p_B$ are the probabilities associated with those features, respectively.
When $p_A$ and $p_B$ are uniform distributions, the output assignment is denoted as $\mathrm{OT}(F_A', F_B')$, which differs from the assignment of matching layer in~\cite{sarlinSuperGlueLearningFeature2020a} by a constant factor.

For matching function with dual-softmax, denote $z_{i,j}$ as $\exp (\mathbf{S}_{i,j})$. 
The assignment $\mathrm{DS}^{p_A, p_B}(F_A', F_B')$ of reweighted dual-softmax is defined as:
\begin{equation}\label{eqn:ds_matching}
  \mathrm{DS}_{i,j} (\mathbf{S}) = \frac{p_A(i) p_B(j) z_{i,j}^2}{{\sum\limits_{k = 1}^{{n_A}} p_A(k){z_{i,k}} \sum\limits_{l = 1}^{{n_B}} p_B(l){z_{l,j}}}} . 
\end{equation}

\begin{theorem}\label{thm:rw_matching}
Let $F_A'$ and $F_B'$ be two sequences of i.i.d. feature points which are sampled from $F_A^{*'}$ and $F_B^{*'}$ with detection probabilities $p_A$ and $p_B$, respectively.
For any index pair $i^*$ and $j^*$ of $F_A^{*'}$ and $F_B^{*'}$, as the lengths of $F_A'$ and $F_B'$ tend to infinity, we have
\begin{equation}
  \sum_{\substack{k,l: \\ k^*=i^*, l^*=j^*}}\mathrm{OT}_{k,l}(F_A', F_B') \stackrel{P}{\longrightarrow} \mathrm{OT}_{i^*,j^*}^{p_A, p_B}(F_A^{*'}, F_B^{*'}).
\end{equation}

\begin{equation}
  \sum_{\substack{k,l: \\ k^*=i^*, l^*=j^*}}\mathrm{DS}_{k,l}(F_A', F_B') \stackrel{P}{\longrightarrow} \mathrm{DS}_{i^*,j^*}^{p_A, p_B}(F_A^{*'}, F_B^{*'}).
\end{equation}
\end{theorem}
% In sparse matching, there are $n_A \times n_B$ potential matches. 
% In dense matching, there are $n_A^* \times n_B^*$ matches, where $n_A^*$ and $n_B^*$ are cardinalities of $F_A^*$ and $F_B^*$ respectively.
Consequently, the detector-based matching network, after reweighting its attention and matching functions based on detection probabilities, can generalize to dense features under aligned marginal distributions.
The proof of the theorems can be found in the Appendix.

\subsection{Sparsity Transfer}\label{sec:sparsity_transfer_method}
Utilizing the reweighting methods, existing matching networks can be transferred to sparsity levels distinct from the training sparsity levels.

\subsubsection{Transferring Detector-Based Matchers}
To enable reweighted dense matching using existing detector-based matchers, it is necessary to indentify feature points $F_I^*$ and their associated probability scores for reweighting. 

Most detectors apply non-maximum suppression (NMS) to the score maps.
We follow the NMS setting of existing detectors and consider features that retain positive scores after NMS as viable features.
Dense feature points are obtained by simply increasing the number of selected top-k features to match the count of features in the feature map.
For instance, given an image of $640\times480$ pixels with an 8-fold downsampled feature map, we construct $F_I^*$ with 4800 feature points.
If the feature map is excessively large, we use fewer features to stay within memory constraints, a practice termed semi-dense matching.

The probability scores are set to the values of score map at corresponding localizations of $F_I^*$.
These scores are utilized to reweight the matcher as described in \cref{sec:rw_attention}.
It is noteworthy that although the score maps of existing detectors are often regarded as probabilistic scores, these detectors are not trained with the probabilistic sampling outlined in our theoretical framework.
SuperPoint~\cite{detoneSuperPointSelfSupervisedInterest2018a} does not employ random sampling of feature points during training. 
Consequently, its score map lacks a strict probabilistic interpretation. 
DISK~\cite{tyszkiewiczDISKLearningLocal2020a} employs random sampling during training, but the distribution range of each sampled point is limited to predefined image patches rather than the entire image.
The analysis of their actual reweighted matching performance is presented in \cref{sec:relative_pose}.

\subsubsection{Transferring Detector-Free Matchers}
To enable detector-free matchers to perform sparse matching without compromising their dense matching capability, we introduce a straightforward sparse training pipeline based on the reweighting method. 
Additionally, a pruning and reweighting pipeline is employed to accelerate inference, as illustrated in \cref{fig:sparsity_transfer}. 

\begin{figure}[tb]
  \centering
  \includegraphics[width=\linewidth, trim={0.0cm, 1.0cm, 2.0cm, 0.5cm}]{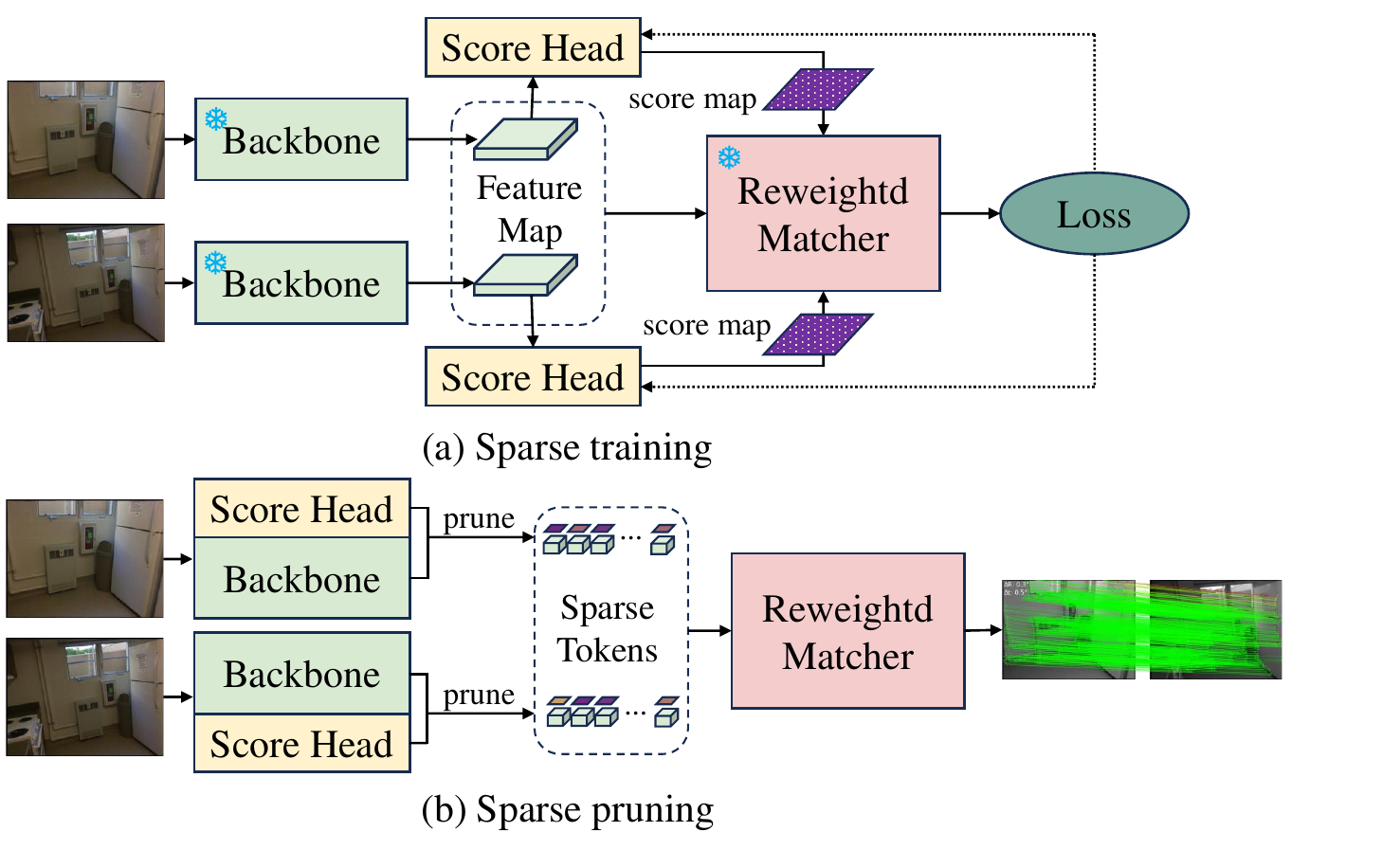}%
  \caption{Sparse training and pruning of detector-free matcher.
  Using reweighted attention and matching function, we employ:
  (a) Sparse training with frozen backbone and matcher.
  (b) Sparse pruning for efficient inference.
  }
  \label{fig:sparsity_transfer}%
\end{figure}

A shallow score head, which is composed of two convolutional layers, is introduced between the backbone and the attention-based matching network. 
The score head is trained while the parameters of the backbone and matching network are frozen to retain the dense matching capability. 
During training, all features participate in the matching process, enabling the score head to learn the global importance of each feature. 
During inference, we prune backbone features with scores below a predefined threshold and feed the remaining features as tokens into the matching network.

For each feature map, the score head computes a score map $\mathbf{S}$ assigning each feature a score between 0 and 1, representing the probability $P(x|I,\theta_F)$ in \cref{sec:preliminaries}. 
The matching network then performs reweighted attention and matching based on the feature map and the score map as in \cref{sec:rw_attention}. 
The original dense matching is equivalent to the reweighted matching when the score head outputs a constant value, in which case the equivalent distribution $P(x|I,\theta_F)$ is uniform distribution.
The objective of the training is to identify a sparse distribution that maximizes the matching performance under the reweighted framework. 
Therefore, the total loss consists of a sparsity loss $L_s$ and a matching loss $L_m$, i.e., $L = L_m + \lambda L_s$, where $\lambda$ is a hyperparameter controlling the level of sparsity.
We adopt $L_1$ loss as sparsity loss:
\begin{equation}
  L_s = \left\lVert \mathbf{S} \right\rVert _ 1,
\end{equation}
which progressively induces sparsity in $\mathbf{S}$ during the training process.
For sparse training of LoFTR, we use the coarse matching supervision as in~\cite{sunLoFTRDetectorFreeLocal2021}:

\begin{equation}
  L_m = -\frac{1}{| \mathcal{M}^{gt} |} \sum_{(i,j) \in \mathcal{M}^{gt}} \log \mathcal{P}(i,j).
\end{equation}

The matching loss computes negative log-likelihood over the ground truth matches $\mathcal{M}^{gt}$, and the matching score $\mathcal{P}$ is computed by the reweighted dual-softmax matching function as in~\ref{eqn:ds_matching}.
We separately train score heads for the indoor and outdoor versions of LoFTR, maintaining consistency with the training sets and parameters provided by the original authors. For each version of LoFTR, we train two distinct score heads with different levels of sparsity.

\section{Experiments}

We implement reweighted versions of mainstream Transformer-based matchers including SuperGlue, LightGlue, and LoFTR.
Their performances are evaluated at different levels of sparsity. 
The comparisons of the sparse and dense matching performance are conducted on both indoor and outdoor datasets.

\subsection{Relative Pose Estimation}\label{sec:relative_pose}

We evaluate the relative pose estimation capability of feature matchers at different levels of sparsity.
ScanNet~\cite{daiScanNetRichlyAnnotated3D2017} and MegaDepth~\cite{liMegadepthLearningSingleview2018} are utilized for indoor and outdoor assessment, respectively.

\subsubsection{Settings}
ScanNet comprises image pairs with wide baselines and texture-less regions. 
We use 1500 test image pairs from ScanNet as in~\cite{sarlinSuperGlueLearningFeature2020a}, with all images resized to a resolution of $640\times480$.

MegaDepth contains image pairs featuring drastic viewpoint changes and texture repetitions. 
We adhere to the testing procedure outlined in~\cite{lindenbergerLightglueLocalFeature2023a} and utilize 1500 test image pairs from St. Peters Square and Reichstag of MegaDepth for evaluation. 
For SuperGlue and LightGlue, the longest side of the input images is resized to 1600 pixels.
For LoFTR, the longest side is resized to 840 pixels.

To reweight SuperGlue and LightGlue, we use official models with their original parameters intact, reweighting only the attention scores based on the probabilistic scores of features at inference time, as in \cref{eq:reweighted_attention}. 
The reweighted SuperGlue and LightGlue utilize feature points and probabilistic scores from SuperPoint or DISK. 
Additionally, the pruning mechanisms of LightGlue are turned off for optimal accuracy.
To reweight LoFTR, we train multiple score heads on indoor/outdoor scenes with varying sparsity.
The backbone features of LoFTR are pruned based on the sparse probabilistic scores derived from the score heads.
The pruned features are then matched based on the probabilistic scores, and its attention scores are also reweighted as in \cref{eq:reweighted_attention}. 

The sparsity of matching is varied by adjusting the number of features used for matching. 
The sparsity of SuperGlue and LightGlue is represented by the number of feature points utilized, while the sparsity of reweighted LoFTR is determined by the quantity of features selected for sparse matching.

We benchmark sparse matching performance by comparing our reweighted LoFTR with established sparse matchers (SuperGlue and LightGlue).
On ScanNet, all baseline methods use 1024 features per image, which follows ~\cite{sarlinSuperGlueLearningFeature2020a}.
On MegaDepth, all baseline methods adhere to the configuration specified in ~\cite{lindenbergerLightglueLocalFeature2023a}, employing 2048 features per image.
And we evaluate reweighted LoFTR with sparsity levels comparable to baseline methods.

For evaluation of dense and semi-dense matching performance, we compare the performance of reweighted SuperGlue and LightGlue against their sparse baselines. 
On ScanNet, we consider dense matching on feature maps at 1/8 of the input resolution and extract 4800 feature points with the highest probabilistic scores as dense features. 
On MegaDepth, we extract 12288 feature points with the highest probabilistic scores as semi-dense features. 

The pose error is computed as the maximum angular error in rotation and translation and reported as Area Under the Curve (AUC) at thresholds of 5°, 10°, and 20°. 
To obtain camera poses, we solve the essential matrix with vanilla RANSAC and LO-RANSAC as in ~\cite{lindenbergerLightglueLocalFeature2023a}.  

\subsubsection{Results of Indoor Pose Estimation}
The results are shown in \cref{table:scannet}, where methods marked with (R) indicate those that are reweighted.
At similar sparsity levels, reweighted LoFTR surpasses the performance of LightGlue and attains results comparable to those of SuperGlue.
When utilizing only half of the features, reweighted LoFTR achieves performance comparable to that of LightGlue under both 10° and 20° conditions.
It is noteworthy that the matching component of LoFTR is trained on dense features, whose parameters remain unchanged during the sparse reweighting process. 
This indicates that for indoor data, training with dense features enables the matching networks to acquire prior knowledge of scene structures, which in turn assists in their generalization when applied to sparse features.
Although detector-based matchers exhibit inferior performance compared to LoFTR in dense matching, the reweighting method enables them to adapt to dense features and improves pose estimation accuracy. 
Among detector-based matchers, reweighted SuperGlue attains the highest performance when utilizing dense SuperPoint features.

\begin{table}
\caption{Evaluation for indoor pose estimation at different sparsity. Bold denotes optimal across all sparsity levels. Underline indicates best using features less than 4000.}
  \centering
\resizebox{\linewidth}{!}{

\begin{tabular}{lccccccc}
\toprule
{\multirow{2}{*}{Method}} & \multirow{2}{*}{Features} & \multicolumn{3}{c}{RANSAC AUC} & \multicolumn{3}{c}{LO-RANSAC AUC} \\
\cmidrule(lr){3-8}
                          &                           &                 \multicolumn{6}{c}{ 5° / 10° / 20°}            \\
\midrule
SP+SuperGlue              &   \multirow{3}{*}{1024}      & \multicolumn{3}{c}{15.60 / 32.73 / 51.19} & \multicolumn{3}{c}{20.03 / 39.01 / \underline{57.55}} \\
SP+LightGlue              &                            & \multicolumn{3}{c}{13.37 / 28.36 / 44.83} & \multicolumn{3}{c}{18.65 / 35.25 / 52.08} \\
DISK+LightGlue            &                            & \multicolumn{3}{c}{13.33 / 26.37 / 40.52} & \multicolumn{3}{c}{17.08 / 31.35 / 46.04} \\
\midrule
\multirow{2}{*}{LoFTR(R)} &            516               & \multicolumn{3}{c}{11.76 / 27.10 / 45.84} & \multicolumn{3}{c}{16.07 / 34.14 / 53.52} \\
                          &           1060               & \multicolumn{3}{c}{\underline{\textbf{16.79}} / \underline{33.61} / \underline{51.38}} & \multicolumn{3}{c}{\underline{\textbf{20.45}} / \underline{39.55} / 57.00} \\
\midrule
SP+SuperGlue(R)   & \multirow{3}{*}{4800}    & \multicolumn{3}{c}{15.82 / \textbf{34.09} / \textbf{54.32}} & \multicolumn{3}{c}{19.51 / \textbf{40.01} / \textbf{59.70}} \\
SP+LightGlue(R)   &                         & \multicolumn{3}{c}{13.94 / 28.80 / 44.81} & \multicolumn{3}{c}{18.79 / 36.11 / 52.55} \\
DISK+LightGlue(R) &                         & \multicolumn{3}{c}{15.54 / 30.35 / 44.64} & \multicolumn{3}{c}{19.02 / 35.36 / 50.77} \\
\bottomrule
\end{tabular}
}

\label{table:scannet}
\end{table}

\subsubsection{Results of Outdoor Pose Estimation}
The outdoor results are presented in \cref{table:megadepth}.
The reweighted LoFTR requires a greater number of features to achieve performance comparable to SuperPoint+LightGlue, attributed to its lower input resolution relative to detector-based methods. 
Nevertheless, with a similar quantity of features, the reweighted LoFTR matches the accuracy of DISK+LightGlue.
Unlike indoor datasets, the MegaDepth dataset is characterized by abundant repetitive textures, which can lead to confusion when numerous features are extracted.
The information gain provided by dense SuperPoint features is insufficient to counteract the impact of feature confusion and generalization issues, thereby degrading the matching performance.
DISK mitigates the feature confusion through a deeper network architecture, enabling reweighted LightGlue to significantly enhance accuracy with a high volume of features.
% The information gain provided by dense matching is insufficient to counteract the impact of feature confusion and generalization issues.

\begin{table}
\caption{Evaluation for outdoor pose estimation at different sparsity. Bold denotes optimal across all sparsity levels. Underline indicates best using features less than 4000.}
  \centering
\resizebox{\linewidth}{!}{

\begin{tabular}{lccccccc}
\toprule
{\multirow{2}{*}{Method}} & \multirow{2}{*}{Features}  & \multicolumn{3}{c}{RANSAC AUC} & \multicolumn{3}{c}{LO-RANSAC AUC} \\
\cmidrule(lr){3-8}
                          &                            &                 \multicolumn{6}{c}{ 5° / 10° / 20°}            \\
\midrule
SP+SuperGlue              &   \multirow{3}{*}{2048}       & \multicolumn{3}{c}{48.74 / 66.09 / 79.52} & \multicolumn{3}{c}{64.55 / 77.60 / 86.87} \\
SP+LightGlue              &                            & \multicolumn{3}{c}{50.03 / 66.98 / 79.90} & \multicolumn{3}{c}{\underline{\textbf{66.80}} / \underline{\textbf{79.19}} / \underline{\textbf{87.84}}} \\
DISK+LightGlue            &                            & \multicolumn{3}{c}{46.65 / 62.98 / 76.37} & \multicolumn{3}{c}{61.28 / 74.40 / 84.03} \\
\midrule
\multirow{2}{*}{LoFTR(R)} &           2110             & \multicolumn{3}{c}{47.34 / 63.56 / 76.29} & \multicolumn{3}{c}{61.37 / 74.23 / 83.58} \\
                          &           2896             & \multicolumn{3}{c}{\underline{\textbf{51.04}} / \underline{67.57} / \underline{80.01}} & \multicolumn{3}{c}{64.53 / 77.09 / 86.02} \\
\midrule
SP+SuperGlue(R)   & \multirow{3}{*}{12288} & \multicolumn{3}{c}{45.30 / 63.15 / 77.44} & \multicolumn{3}{c}{62.17 / 75.91 / 85.62} \\
SP+LightGlue(R)   &                        & \multicolumn{3}{c}{49.25 / 66.51 / 79.35} & \multicolumn{3}{c}{66.08 / 78.75 / 87.47} \\
DISK+LightGlue(R) &                        & \multicolumn{3}{c}{50.47 / \textbf{67.64} / \textbf{80.28}} & \multicolumn{3}{c}{66.10 / 78.05 / 86.41} \\
\bottomrule
\end{tabular}
}

\label{table:megadepth}
\end{table}

\subsection{Visual Localization}

We evaluate long-term visual localization under semi-dense features.
The Aachen Day-Night~\cite{zhangReferencePoseGeneration2021} (outdoor) and the InLoc~\cite{tairaInLocIndoorVisual2021} (indoor) dataset are used for evaluation.

\subsubsection{Settings}
We generate results for all methods using the HLoc\cite{sarlinCoarseFineRobust2019a} toolbox with default settings for triangulation, image retrieval, RANSAC, and Perspective-n-Point solver. 
The original matchers are evaluated as baselines and the feature extraction and matching settings provided by HLoc are adopted. 
The original SuperGlue and LightGlue perform sparse matching, while the original LoFTR performs semi-dense matching with a maximum number of feature points per image limited to 8192. 
To transfer the detector-based matchers to a sparsity level comparable to LoFTR, we increase the number of feature points per image to 8192 and reweight the matcher.

\subsubsection{Results}
Outdoor evaluation results are shown in \cref{table:aachen}. 
For daytime data, the reweighted semi-dense matchers perform similarly to the original matchers. 
For nighttime data, the reweighted semi-dense matchers with SuperPoint outperform their sparse counterparts. 
LoFTR is less accurate than sparse matchers on daytime data but is more accurate on nighttime data. 
Intriguingly, after being transferred to semi-dense features, detector-based matchers demonstrate better generalization than LoFTR on nighttime data.

Results of indoor data are presented in \cref{table:inloc}. 
Similar to relative pose estimation, reweighted semi-dense matching achieves higher performance gains on indoor data compared to outdoor data. 
More stable improvements are observed at lower distance thresholds, with reweighted SuperGlue surpassing LoFTR at (0.25 m, 10°).

\begin{table}
\caption{Evaluation for visual localization on the Aachen Day-Night~\cite{zhangReferencePoseGeneration2021} benchmark v1.1.
Methods marked with (R) indicate reweighted semi-dense matching. The best performance is marked in bold.}
  \centering
\resizebox{\linewidth}{!}{

\begin{tabular}{lcccccc}
\toprule
\multicolumn{1}{l}{\multirow{2}{*}{Method}} & \multicolumn{3}{c}{Day} & \multicolumn{3}{c}{Night}\\
\cmidrule(lr){2-7}
                  &                 \multicolumn{6}{c}{ (0.25m,2°) / (0.5m,5°) / (1.0m,10°) }    \\
\midrule
SP+SuperGlue      & \multicolumn{3}{c}{ \textbf{90.4} / \textbf{96.5} / 99.3 } & \multicolumn{3}{c}{ 76.4 / 91.1 / \textbf{100.0} } \\
SP+SuperGlue(R)   & \multicolumn{3}{c}{ 90.3 / 96.2 / 99.2 } & \multicolumn{3}{c}{ \textbf{77.5} / 91.1 / 99.5 } \\

SP+LightGlue      & \multicolumn{3}{c}{ 90.3 / 96.1 / 99.3 } & \multicolumn{3}{c}{ 77.0 / 91.6 / \textbf{100.0} } \\
SP+LightGlue(R)   & \multicolumn{3}{c}{ \textbf{90.4} / \textbf{96.5} / 99.2 } & \multicolumn{3}{c}{ \textbf{77.5} / \textbf{92.1} / \textbf{100.0} } \\

DISK+LightGlue    & \multicolumn{3}{c}{ 88.1 / 95.5 / 99.3 } & \multicolumn{3}{c}{ 76.4 / 90.1 / 99.5 } \\
DISK+LightGlue(R) & \multicolumn{3}{c}{ 88.6 / 96.1 / \textbf{99.4} } & \multicolumn{3}{c}{ 74.9 / 90.6 / 99.5 } \\

LoFTR             & \multicolumn{3}{c}{ 88.8 / 95.8 / 98.8 } & \multicolumn{3}{c}{ \textbf{77.5} / 91.1 / 99.5 } \\
\bottomrule
\end{tabular}
}

\label{table:aachen}
\end{table}

\begin{table}
\caption{Evaluation for visual localization on the InLoc~\cite{tairaInLocIndoorVisual2021} benchmark.
Methods marked with (R) indicate reweighted semi-dense matching. The best performance is marked in bold.}
  \centering
\resizebox{\linewidth}{!}{

\begin{tabular}{lcccccc}
\toprule
\multicolumn{1}{l}{\multirow{2}{*}{Method}} & \multicolumn{3}{c}{DUC1} & \multicolumn{3}{c}{DUC2}\\
\cmidrule(lr){2-7}
                  &                 \multicolumn{6}{c}{ (0.25m,10°) / (0.5m,10°) / (1.0m,10°) }    \\
\midrule
SP+SuperGlue      & \multicolumn{3}{c}{ 46.0 / 65.7 / 79.3 } & \multicolumn{3}{c}{ 47.3 / 71.0 / 77.1 } \\
SP+SuperGlue(R)   & \multicolumn{3}{c}{ \textbf{47.5} / 69.2 / 80.8 } & \multicolumn{3}{c}{ \textbf{53.4} / 67.9 / 74.8 } \\

SP+LightGlue      & \multicolumn{3}{c}{ 42.9 / 64.1 / 76.3 } & \multicolumn{3}{c}{ 42.0 / 66.4 / 72.5 } \\
SP+LightGlue(R)   & \multicolumn{3}{c}{ 43.9 / 67.7 / 80.3 } & \multicolumn{3}{c}{ 48.9 / 69.5 / 77.1 } \\

DISK+LightGlue    & \multicolumn{3}{c}{ 42.4 / 59.6 / 73.7 } & \multicolumn{3}{c}{ 36.6 / 56.5 / 68.7 } \\
DISK+LightGlue(R) & \multicolumn{3}{c}{ 43.4 / 64.1 / 77.8 } & \multicolumn{3}{c}{ 38.9 / 55.7 / 65.6 } \\

LoFTR             & \multicolumn{3}{c}{ 46.0 / \textbf{69.7} / \textbf{82.3} } & \multicolumn{3}{c}{ 48.9 / \textbf{74.0} / \textbf{81.7} } \\
\bottomrule
\end{tabular}
}

\label{table:inloc}
\end{table}

\subsection{Analysis of Reweighted SuperGlue}
In \cref{sec:relative_pose}, reweighted dense SuperGlue exhibits superior performance compared to other detector-based methods on ScanNet. 
To analyze the effect of the reweighting, we compare the performance of reweighted SuperGlue against their original counterparts under the same dense features. 
On ScanNet, we employ the original indoor version of SuperGlue to match dense SuperPoint features, referred to as direct dense matching. 
On the MegaDepth dataset, we use the original outdoor version of SuperGlue to match semi-dense SuperPoint features, termed direct semi-dense matching. 
For both direct and reweighted matching on the same dataset, the same feature point settings and matcher parameters are utilized.

As shown in \cref{table:dense_compare}, the reweighted matching achieves higher pose accuracy than the direct matching method on both indoor and outdoor datasets. 
The reweighted matching better leverages the same sets of feature points, indicating superior generalization to dense and semi-dense features for pose estimation. 
The average numbers of matches output by the matchers for each image pair are also recorded.
Interestingly, despite the superior pose estimation performance of the reweighting method, it yields a lower average number of matches compared to direct dense matching.

\begin{table}
  \caption{Comparison with direct dense matching}
    \centering
\resizebox{\linewidth}{!}{

\begin{tabular}{ccccccccc}
\toprule
\multirow{2}{*}{Dataset} & \multicolumn{1}{c}{\multirow{2}{*}{Method}} & \multicolumn{3}{c}{RANSAC AUC} & \multicolumn{3}{c}{LO-RANSAC AUC} & \multirow{2}{*}{Matches} \\
\cmidrule(lr){3-5} \cmidrule(lr){6-8}
                          &                                             & @5° & @10° & @20°     & @5° & @10° & @20° &      \\
\midrule
\multirow{2}{*}{ScanNet}
                          & Direct    & 13.94 & 31.76 & 51.11 & 17.91 & 37.44 & 56.02 & 1106 \\
                          & Reweighted& \textbf{15.82} & \textbf{34.09} & \textbf{54.32} & \textbf{19.51} & \textbf{40.01} & \textbf{59.70} & 737 \\
\midrule
\multirow{2}{*}{MegaDepth}
                          & Direct     & 43.74 & 62.11 & 76.90 & 61.10 & 75.27 & 85.03 & 3518 \\
                          & Reweighted & \textbf{45.30} & \textbf{63.15} & \textbf{77.44} & \textbf{62.17} & \textbf{75.91} & \textbf{85.62} & 2705 \\
\bottomrule
\end{tabular}
}

  \label{table:dense_compare}
\end{table}

To provide a more intuitive analysis of the effect of the reweighting method on the matches, qualitative results are presented in \cref{fig:compare} and \cref{fig:superglue_compare}. 
As illustrated in~\cref{fig:compare} (b), matchers pretrained on sparse features encounter generalization issues when directly matching dense features, resulting in larger pose errors. 
Without additional training, the reweighted SuperGlue generalize better to dense features and achieves lower pose errors in \cref{fig:compare} (c).
The performance enhancement of the reweighted SuperGlue on dense features is primarily attributed to the increased number of matches, which effectively reduces noise in pose estimation.
The results in~\cref{fig:superglue_compare} demonstrate that the direct matching of SuperGlue generates a higher number of matches on feature points of simple image pairs but fewer matches on challenging image pairs, leading to a higher average match count yet a lower pose AUC. 
In contrast, the reweighted SuperGlue exhibits a more balanced quantity of matches and demonstrates greater robustness on difficult image pairs.
As observed in panels (b) and (c) of~\cref{fig:compare}, the improvements of reweighted LightGlue are mainly achieved by better distributed matches.
Although the point pruning in LightGlue is deactivated during the experiment, neither its direct nor reweighted dense matching demonstrates a substantial improvement in match counts compared to reweighted SuperGlue.
This phenomenon may be due to LightGlue's matchability score for each feature point beyond the matching score, which implicitly leads to a sparsification of matches.
We therefore conclude that SuperGlue exhibits stronger length generalization capability than LightGlue.
% This suggests that the information gain from additional features outweighs the detrimental effects stemming from variations in feature distribution. 

\begin{figure*}[!t]
  \centering
  \includegraphics[width=0.9\linewidth,trim={0.5cm, 0.0cm, 0.0cm, 0.5cm}]{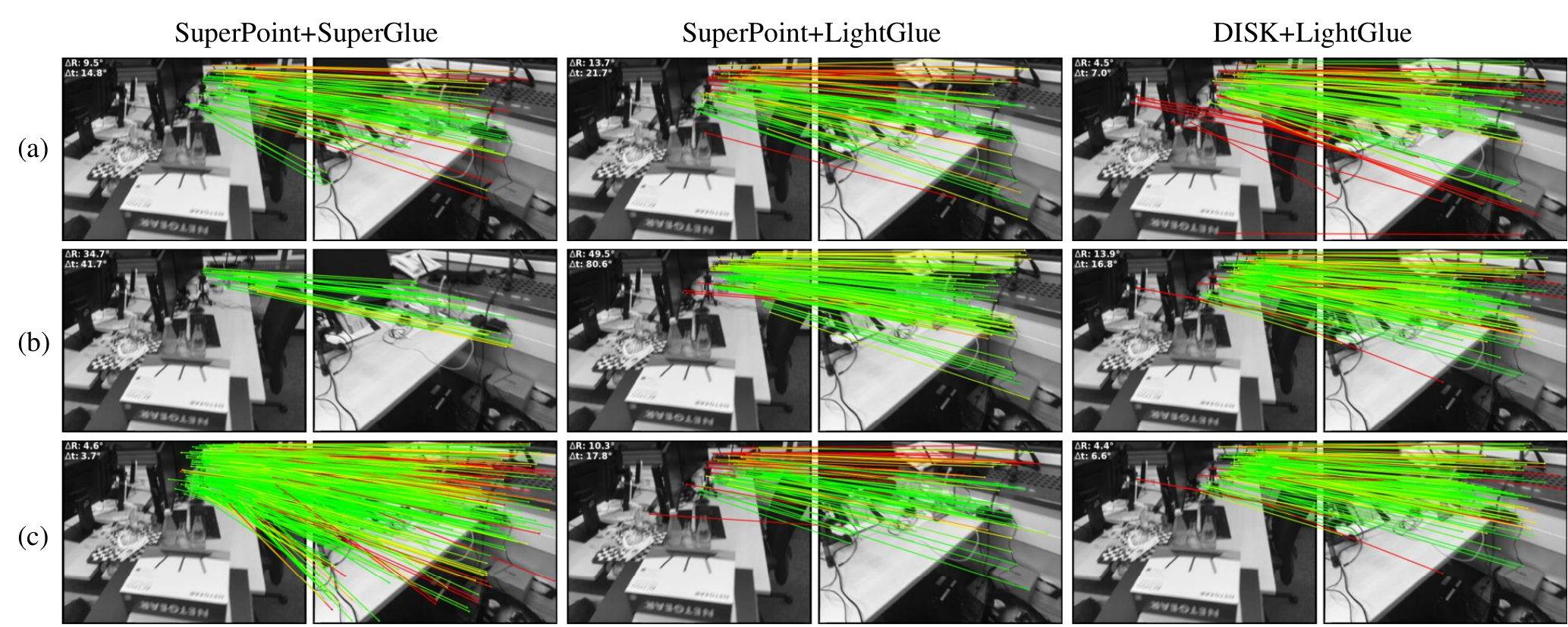}%
  \caption{Results of detector-based matchers under different matching modes.
  (a) Results of direct matching on sparse feature points.
  (b) Results of direct matching on dense feature points.
  (c) Results of reweighted matching on the same dense feature points. 
  }
  \label{fig:compare}
\end{figure*}

\begin{figure}[tb]
  \centering
  \includegraphics[width=0.75\linewidth]{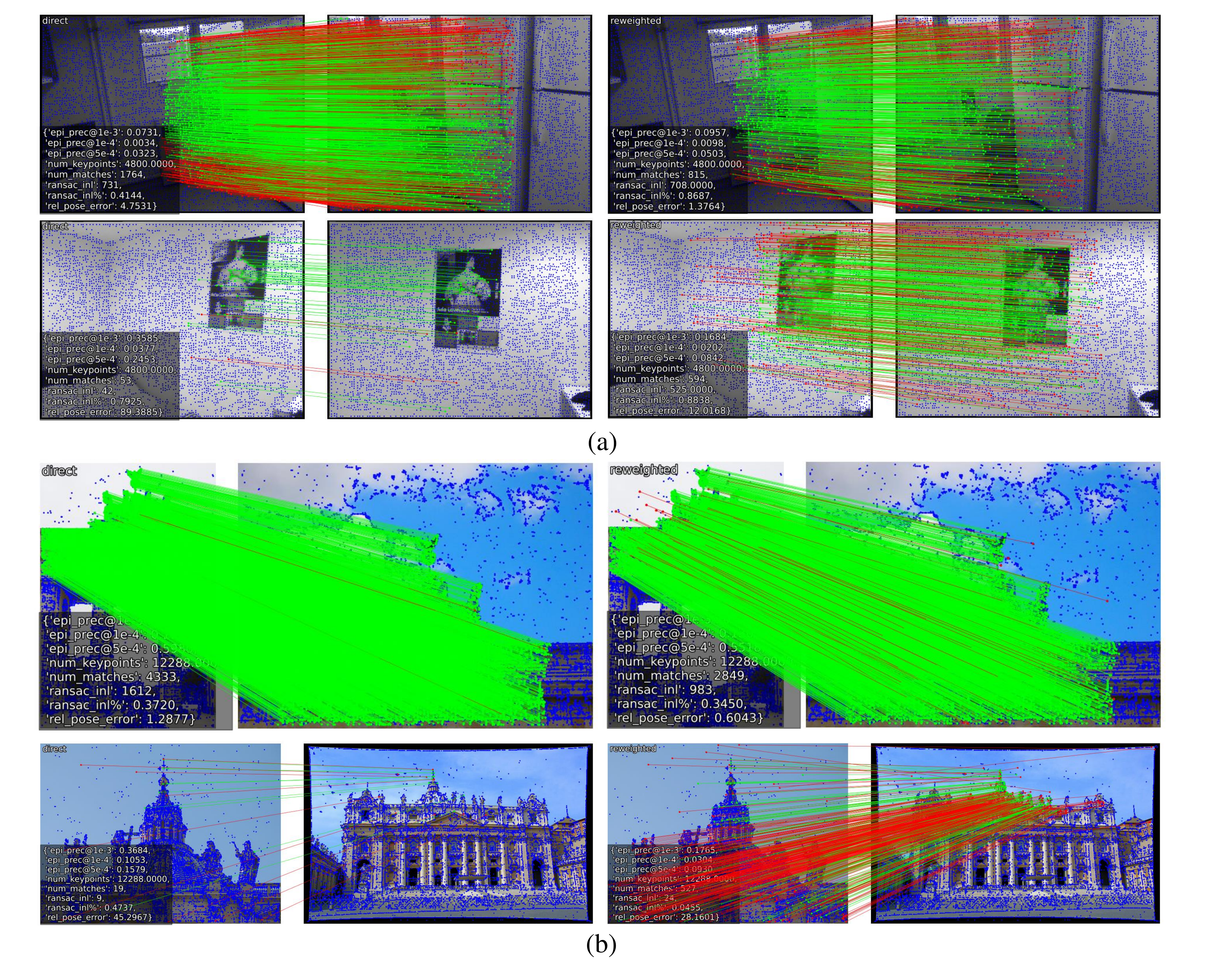}%
  \caption{Results of original and reweighted SuperGlue on image pairs with different difficulties‌. 
  (a) Simple (top) and challenging (bottom) indoor image pairs.
  (b) Simple (top) and challenging (bottom) outdoor image pairs. 
  The left column displays the results of direct matching, while the right column presents the results of reweighted matching.
  }
  \label{fig:superglue_compare}%
\end{figure}

\subsection{Analysis of Sparse LoFTR}

We further analyze the computational cost and accuracy of reweighted LoFTR evaluated in \cref{sec:relative_pose} with varying sparsity levels. 

Through the pipeline in \cref{sec:sparsity_transfer_method}, we conduct sparse training on score heads for the indoor and outdoor versions of LoFTR, separately.
Each score head consists of two convolutional layers and takes a coarse feature map as input.
The coarse matching module of LoFTR is reweighted during sparse training and pruning.
Then we evaluate the networks at varying sparsity levels using different score heads on corresponding test sets and record the feature proportion, FLOPs, and MACs after sparse pruning. 

The results are presented in \cref{table:sparsity_analysis}. 
It can be observed that, due to the use of linear attention, its computational cost is proportional to its sparsity. 
The matcher trained on ScanNet achieves higher sparsity compared to that trained on MegaDepth, attributed to differences in texture richness between the datasets. 
The AUC decay on ScanNet is less significant for low thresholds, suggesting that sparse matching better preserves accuracy on simpler image pairs. 
In contrast, the AUC decay on MegaDepth is consistent across different thresholds for sparse matching.
The score maps and corresponding qualitative results are shown in \cref{fig:loftr_sparse}.
When the proportion of retained features is higher, the score head behaves like an edge detector. 
When the proportion is lower, the score maps undergo further reduction and the score head acts like a point detector.
As the sparsity level varies, sparse LoFTR transfers between edge-based and point-based matching.

\begin{table}
  \caption{Computational cost and accuracy of LoFTR at different levels of sparsity.}
    \centering
\resizebox{\linewidth}{!}{

\begin{tabular}{ccccccc}
\toprule
\multirow{2}{*}{Dataset} & \multirow{2}{*}{Proportion} & \multirow{2}{*}{GFLOPs} & \multirow{2}{*}{GMACs} & \multicolumn{3}{c}{Pose estimation AUC} \\
\cmidrule(lr){5-7}
                          &                            &                        &                       & @5° & @10° & @20°     \\
\midrule
\multirow{2}{*}{ScanNet}
                          & 1.00    & 103.5   & 50.3    & 22.06 & 40.80 & 57.62 \\
                          & 0.22    & 22.7    & 11.0    & 14.46 & 30.26 & 47.10 \\
                          & 0.11    & 11.1    & 5.4     & 7.43  & 18.82 & 35.41 \\
\midrule
\multirow{2}{*}{MegaDepth}
                          & 1.00    & 237.8   & 115.6   & 52.80 & 69.19 & 81.18 \\
                          & 0.35    & 83.3    & 40.5    & 45.31 & 61.56 & 73.83 \\
                          & 0.26    & 60.7    & 29.5    & 41.48 & 57.36 & 70.24 \\ 
\bottomrule
\end{tabular}
}

  \label{table:sparsity_analysis}
\end{table}

\begin{figure}[h]
  \centering
  \includegraphics[width=0.8\linewidth]{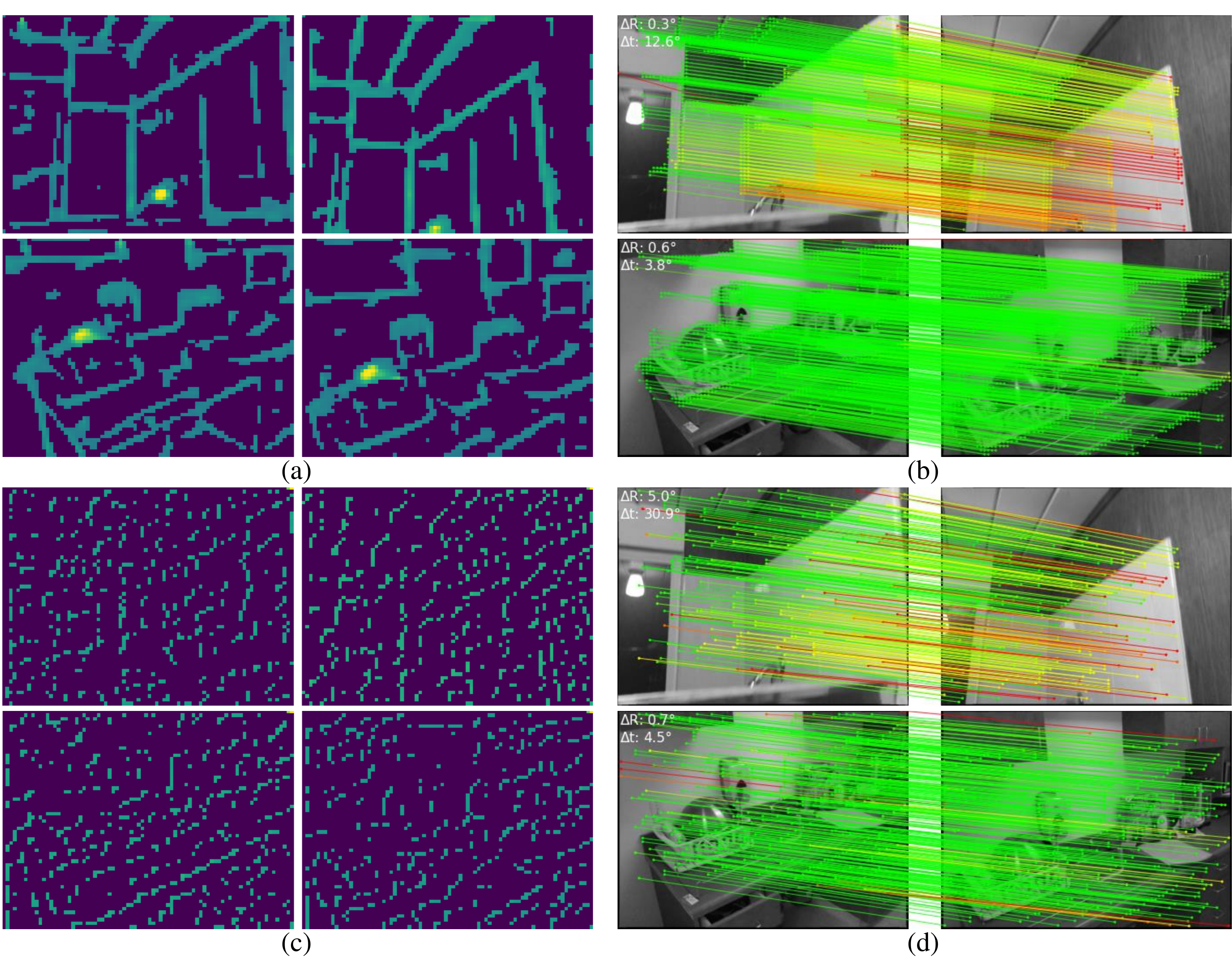}%
  \caption{Matching results of sparse LoFTR.
  (a) and (b) present the score maps and matching results for an average feature proportion of 0.22. 
  (c) and (d) are the score maps and matching results for an average feature proportion of 0.11. 
  % In the score maps, the lighter sections represent the features retained after pruning. 
  }
  \label{fig:loftr_sparse}%
\end{figure}

\begin{figure}[h!]
  \centering
  \includegraphics[width=0.8\linewidth,trim={0.0cm, 0.0cm, 0.0cm, 1.5cm}]{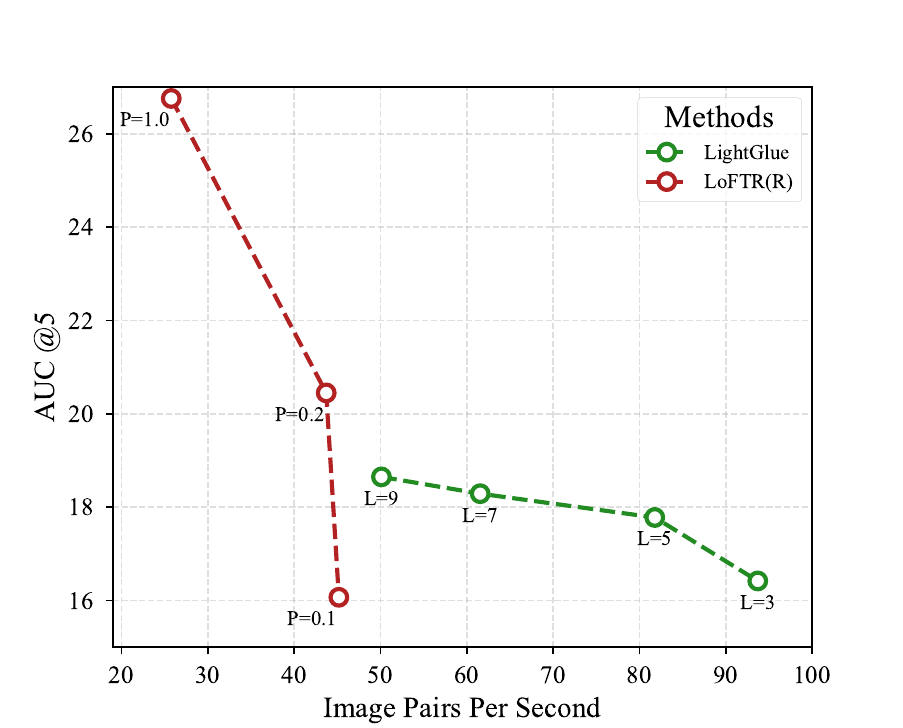}%
  \caption{Speed-Accuracy Trade-off on ScanNet.
    $P$: Proportion of retained features in reweighted LoFTR. $L$: Depth of pruned matching layers in LightGlue.
  }
  \label{fig:pareto}%
\end{figure}

To analyze the efficiency-accuracy trade-off of the reweighted LoFTR, we evaluate its speed and pose accuracy under varying sparsity levels and compared it with LightGlue's depthwise pruning approach. 
The matching performance was tested on the ScanNet dataset using an RTX3090 GPU, maintaining a VGA resolution of $640\times480$ for the images. 
The average frame rate for the matching process is illustrated in~\cref{fig:pareto}.
The complementary nature of sparse reweighting and depth-pruning methods can be observed from the figure. 
The reweighted LoFTR can enhance matching speed by reducing the proportion of features.
However, when the feature proportion falls below 0.2, the marginal speed gain from further reducing the number of features does not compensate for the associated accuracy loss. 
Conversely, LightGlue can improve performance by increasing the number of network layers, but the marginal accuracy gain from adding more layers does not sufficiently offset the speed degradation. 
Future matchers could orthogonally integrate these two strategies to enhance overall efficiency and performance.

\section{Conclusion}
This paper presents a novel reweighting method, which is applicable to a wide range of Transformer-based matching networks and adapts them to distinct sparsity levels without altering network parameters. 
We prove that the reweighted matching network constitutes the asymptotic limit of the detector-based matching network, leading to consistency in probability distribution and aiding generalization to dense features.
Our sparse training and pruning pipeline enables sparse matching for detector-free matchers and preserves their dense performance.
Experiments show that the reweighting method can improve the pose accuracy of SuperGlue and LightGlue on dense indoor feature points.
SuperGlue exhibits stronger length generalization capability than LightGlue.
The reweighted sparse LoFTR achieves comparable accuracy to sparse baselines and offers efficiency-accuracy trade-offs complementary to LightGlue.
\clearpage

{\appendices
\appendices \label{sec:proof}

\section*{Proof of \texorpdfstring{\cref{thm:reweighted_attention_layer}}{theorem on asymptotic equivalence of reweighted attention layer}}
In \cref{sec:rw_attention}, we demonstrate the asymptotic equivalence of the reweighted Transformer and the original Transformer.
To prove this asymptotic equivalence, we first prove \cref{thm:reweighted_attention_layer} concerning the asymptotic equivalence of single attention layer.
Then we prove \cref{thm:rw_attention} for Transformers of arbitrary depth.
Similarly, \cref{thm:rw_matching} can be proved using analogous techniques.

\begin{proof}
    For fixed $m,n$, we omit the superscripts with them and rewrite the computation of the attention layer as follows:
    \begin{equation}
        \begin{split}
        & \mathbf{Y}  = \mathbf{Q} + \sum_{h = 1}^{H} \mathbf{W}_{VO}^{h} \mathbf{V} \mathrm{Sim}(\mathbf{W}_K^{h}\mathbf{K}, \mathbf{W}_Q^{h}\mathbf{Q}) \\
        & \mathbf{Q}' = \mathbf{Y} + \mathbf{W}_2 \mathrm{Act} (\mathbf{W}_1 \mathbf{Y} + \mathbf{b}_1 \mathbf{1}^\top) + \mathbf{b}_2 \mathbf{1}^\top ,
        \end{split}
    \end{equation}
    where $H$ is the number of attention heads, and $\mathbf{W}_{VO}^{h}$, $\mathbf{W}_K^{h}$, and $\mathbf{W}_Q^{h}$ are parameter matrices of the $h$-th attention head. 
    The $\mathbf{W}_1, \mathbf{W}_2$ and $\mathbf{b}_1, \mathbf{b}_2$ are parameter matrices and vectors of the feed forward part.

    Denote the result of the $h$-th attention head as $\mathbf{Y}^h$, and denote the result of the $h$-th reweighted attention head as $\mathbf{Y}^{h*}$:
    \begin{equation}
        \begin{split}
        & \mathbf{Y}^h = \mathbf{W}_{VO}^{h} \mathbf{V} \mathrm{Sim}(\mathbf{W}_K^{h}\mathbf{K}, \mathbf{W}_Q^{h}\mathbf{Q}) \\
        & \mathbf{Y}^{h*} = \mathbf{W}_{VO}^{h} \mathbf{V}^* \mathrm{Sim}^p(\mathbf{W}_K^{h}\mathbf{K}^*, \mathbf{W}_Q^{h}\mathbf{Q}^*).
        \end{split}
    \end{equation}

    For $i \in [m], j \in [n]$, denote $s_{i,j}$ as $\delta(\mathbf{W}_K^{h}\mathbf{K}_i, \mathbf{W}_Q^{h}\mathbf{Q}_j)$,
    and for $i^* \in [m^*]$, $j^* \in [n^*]$, denote $s^*_{i^*,j^*}$ as $\delta(\mathbf{W}_K^{h}\mathbf{K}^*_{i^*}, \mathbf{W}_Q^{h}\mathbf{Q}^*_{j^*})$.
    We have
    \begin{equation} \label{eqn:Yjh}
        \begin{aligned}
           \mathbf{Y}_{j}^{h}
        & =\mathbf{W}_{VO}^h \mathbf{V} \cdot \frac{1}{\sum\limits_{i=1}^{m}{s_{i,j}}}
            \left[ \begin{matrix}
            {{s}_{1,j}}  \\
            {{s}_{2,j}}  \\
            \vdots   \\
            {{s}_{{m},j}}  \\
            \end{matrix} \right] \\ 
        & ={{\mathbf{W}}_{VO}^h}\cdot \frac{1}{\sum\limits_{i=1}^{{m}}{{{s}_{i,j}}}}\sum\limits_{i=1}^{{m}}{{{s}_{i,j}}{{\mathbf{V}}_{i}}} \\ 
        & ={{\mathbf{W}}_{VO}^h}\cdot \frac{1}{ \frac{1}{{m}}\sum\limits_{i=1}^{{m}}{{{s}_{i,j}}} }\cdot \frac{1}{{m}}\sum\limits_{i=1}^{{m}} { {{s}_{i,j}} {{\mathbf{V}}_{i}} }
        .
        \end{aligned}
    \end{equation}

    When $m,n$ approaches infinity, we have
    \begin{equation} \label{eqn:E_s}
        \begin{aligned}
            \frac{1}{m}\sum\limits_{i=1}^{m}{{{s}_{i,j}}} &  =\sum\limits_{{{i}^{*}}=1}^{{{m}^{*}}} \frac{\left| {{\mathcal{I}}_{m}}({{i}^{*}}) \right|}{m} \\
            {} & \cdot \frac{1}{\left| {{\mathcal{I}}_{m}}({{i}^{*}}) \right|}\sum\limits_{i\in {{\mathcal{I}}_{m}}({{i}^{*}})}{\delta (\mathbf{W}_{K}^{h}\mathbf{K}_{i}^{(m)},\mathbf{W}_{Q}^{h}\mathbf{Q}_{j}^{(n)})} \\
            {} & \xrightarrow{P}\sum\limits_{{{i}^{*}}=1}^{{{m}^{*}}}{p({{i}^{*}})\delta (\mathbf{W}_{K}^{h}\mathbf{K}_{{{i}^{*}}}^{*},\mathbf{W}_{Q}^{h}\mathbf{Q}_{{{j}^{*}}}^{*})}  \\
            {} & =\sum\limits_{{{i}^{*}}=1}^{{{m}^{*}}}{p({{i}^{*}})s_{{{i}^{*}},{{j}^{*}}}^{*}},  \\                 
        \end{aligned}
    \end{equation}
    and
    % \begin{equation} \label{eqn:E_sV}
    %     \begin{aligned}
    %         \frac{1}{m}\sum\limits_{i=1}^{m}{{{s}_{i,j}}} {{\mathbf{V}}_{i}^{(m)}}  &  =\sum\limits_{{{i}^{*}}=1}^{{{m}^{*}}} \frac{\left| {{\mathcal{I}}_{m}}({{i}^{*}}) \right|}{m} \\
    %         {} & \cdot  \frac{1}{\left| {{\mathcal{I}}_{m}}({{i}^{*}}) \right|}\sum\limits_{i\in {{\mathcal{I}}_{m}}({{i}^{*}})}{\delta (\mathbf{W}_{K}^{h}\mathbf{K}_{i}^{(m)},\mathbf{W}_{Q}^{h}\mathbf{Q}_{j}^{(n)})}  {{\mathbf{V}}_{i}^{(m)}} \\
    %         {} & \xrightarrow{P} \sum\limits_{{{i}^{*}}=1}^{{{m}^{*}}}{p({{i}^{*}})s_{{{i}^{*}},{{j}^{*}}}^{*}}  {{\mathbf{V}}_{i^*}^*},  \\ 
    %     \end{aligned}
    % \end{equation}
    \begin{equation} \label{eqn:E_sV}
            \frac{1}{m}\sum\limits_{i=1}^{m}{{{s}_{i,j}}} {{\mathbf{V}}_{i}^{(m)}}  
            \xrightarrow{P} \sum\limits_{{{i}^{*}}=1}^{{{m}^{*}}}{p({{i}^{*}})s_{{{i}^{*}},{{j}^{*}}}^{*}}  {{\mathbf{V}}_{i^*}^*}
            . 
    \end{equation}

    Substitute \cref{eqn:E_s,eqn:E_sV} into \cref{eqn:Yjh}, we have
    \begin{equation}
        \begin{aligned}
            \mathbf{Y}_{j}^{h(m,n)}
            & \xrightarrow{P} {{\mathbf{W}}_{VO}^h}\cdot \frac{1}{ \sum\limits_{{{i}^{*}}=1}^{m^*} {p({{i}^{*}})s_{{{i}^{*}},{{j}^{*}}}^{*}} } 
              \sum\limits_{{{i}^{*}}=1}^{m^*} {p({{i}^{*}})s_{{{i}^{*}},{{j}^{*}}}^{*}} \mathbf{V}_{{{i}^{*}}}^{*} \\
            & = {{\mathbf{W}}_{VO}^h} \mathbf{V}^{*}\cdot \frac{1}{ \sum\limits_{{{i}^{*}}=1}^{m^*} {p({{i}^{*}})s_{{{i}^{*}},{{j}^{*}}}^{*}} } 
                \left[ \begin{matrix}
                p(1){{s}_{1,j^*}^*}  \\
                p(2){{s}_{2,j^*}^*}  \\
                \vdots   \\
                p({m^*}){{s}_{{m^*},j^*}^*}  \\
                \end{matrix} \right] \\ 
            & = \mathbf{Y}_{j^*}^{h*}.
        \end{aligned}
    \end{equation}

    Therefore, we have
    \begin{equation}
        \begin{aligned}
            \mathbf{Y}_j^{(m,n)} 
            & = \mathbf{Q}_j^{(n)} + \sum_{h = 1}^{H} \mathbf{Y}_{j}^{h(m,n)} \\
            & \xrightarrow{P} \mathbf{Q^*}_{j^*} + \sum_{h = 1}^{H} \mathbf{Y}_{j^*}^{h*} \\
            & = \mathbf{Y}^*_j \text{.}
        \end{aligned}
    \end{equation}

    Because the feed forward part is continuous element-wise function, we obtain
    \begin{equation}
        \begin{split}
        \mathbf{Q}_{j}^{'(m,n)} 
        & =  \mathbf{Y}_j^{(m,n)} + \mathbf{W}_2 \mathrm{Act} (\mathbf{W}_1 \mathbf{Y}_j^{(m,n)} + \mathbf{b}_1 \mathbf{1}^\top) + \mathbf{b}_2 \mathbf{1}^\top \\
        & \xrightarrow{P} \mathbf{Y}^*_{j^*} + \mathbf{W}_2 \mathrm{Act} (\mathbf{W}_1 \mathbf{Y}^*_{j^*} + \mathbf{b}_1 \mathbf{1}^\top) + \mathbf{b}_2 \mathbf{1}^\top \\
        & = \mathbf{Q}^{'*}_{j^*}
        .
        \end{split}
    \end{equation}
    
\end{proof}

\section*{Proof of \texorpdfstring{\cref{thm:rw_attention}}{theorem on reweighted Transformer}}

\begin{proof}
    As described in \cref{sec:preliminaries}, feature points $F_A$ and $F'_A$ only differ in feature descriptors, keeping points coordinates unchanged. 
    Therefore, we only need to examine the convergence values of the updated feature descriptors, namely, the output tokens of the network.
    We first examine the outputs of the embedding layer. 
    Then we employ mathematical induction on the depth of the Transformers $\mathrm{TF}$ and $\mathrm{TF}^{p_A, p_B}$. 

    The embedding layer is the first layer of the Transformer, and it is point-wise, which aligns with mainstream designs like rotary position embedding.
    Thus, the embedding of each feature point is not influenced by other feature points.
    Let the size of $F_A$ and $F_A^*$ be $n_A$ and $n_A^*$, respectively.
    As described in \cref{sec:preliminaries}, $F_A$ are i.i.d. samples of $F_A^*$, and for any index $i, 1 \le i \le n_A$ and its corresponding index $i^*, 1 \le i^* \le n_A^*$, we have $F_{A,i} = F_{A,i^*}^*$.
    Denote the resulting tokens from the embedding layer of $F_A$ and $F_A^*$ as $\mathbf{X}_A$ and $\mathbf{X}_A^*$, respectively.
    Because we don't reweight the point-wise embedding layer, we have $\mathbf{X}_{A,i} = \mathbf{X}_{A,i^*}^*$.
    Similarly, $\mathbf{X}_{B,j} = \mathbf{X}_{B,j^*}^*, 1 \le j \le n_B, 1 \le j^* \le n_B^*$.

    The attention layers aggregates input tokens. 
    Let the number of attention layers in $\mathrm{TF}$ be $l$. 
    Denote the input tokens of the $l$-th attention layer in $\mathrm{TF}$ as $\mathbf{X}_Q$, $\mathbf{X}_K$, and $\mathbf{X}_V$. 
    Similarly, the input tokens of the $l$-th reweighted attention layer in $\mathrm{TF}^{p_A, p_B}$ are denoted as $\mathbf{X}_Q^*$, $\mathbf{X}_K^*$ and $\mathbf{X}_V^*$.
    In a typical Transformer design, we have $\mathbf{X}_K = \mathbf{X}_V$ and $\mathbf{X}_K^* = \mathbf{X}_V^*$.
    Thus, \cref{eqn:condition_shape} is satisfied.
    According to the law of large numbers, \cref{eqn:ratio_converges_to_probability} is satisfied.
    
    When $l = 1$, $\mathrm{TF}$ consists of one embedding layer and one attention layer. 
    We have $\mathbf{X}_Q, \mathbf{X}_K, \mathbf{X}_V \in \left\{\mathbf{X}_A, \mathbf{X}_B\right\} $, and the exact assignment depends on the design choice of the layer.
    According to the discussion on the embedding layer, \cref{eqn:condition_limQ,eqn:condition_limKV} hold true.
    Applying \cref{thm:reweighted_attention_layer}, we obtain that the output tokens of $\mathrm{TF}$ converge in probability to the output tokens of $\mathrm{TF}^{p_A, p_B}$.

    When $l > 1$, applying mathematical induction on Transformers with $l - 1$ attention layers, \cref{eqn:condition_limQ,eqn:condition_limKV} are satisfied.
    Again, applying \cref{thm:reweighted_attention_layer}, \cref{thm:rw_attention} is proved.
    
\end{proof}

\section*{Proof of \texorpdfstring{\cref{thm:rw_matching}}{theorem on reweighted matching function}}
\begin{proof}
Original optimal transport in image matching assigns equal weights to the points to be matched:
\begin{equation}
    \begin{aligned}
       & {{\mathbf{a}}_{i}}=\frac{1}{{{n}_{A}}},{{\mathbf{a}}_{{{n}_{A}}+1}}=1,1\le i\le {{n}_{A}} \\ 
       & {{\mathbf{b}}_{j}}=\frac{1}{{{n}_{B}}},{{\mathbf{b}}_{{{n}_{B}}+1}}=1,1\le j\le {{n}_{B}}, \\ 
      \end{aligned}      
\end{equation}
where ${n}_{A}, {n}_{B}$ are numbers of sampled features of image $A$ and $B$ which tend to infinity.

Reweighted optimal transport employs feature detection probabilities ${p}_{A}$ and ${p}_{B}$ for the weights:
\begin{equation}
    \begin{aligned}
       & \mathbf{a}_{{{i}^{*}}}^{*}={{p}_{A}}({{i}^{*}}),\mathbf{a}_{n_{A}^{*}+1}^{*}=1,1\le {{i}^{*}}\le n_{A}^{*} \\ 
       & \mathbf{b}_{{{j}^{*}}}^{*}={{p}_{B}}({{j}^{*}}),\mathbf{b}_{n_{B}^{*}+1}^{*}=1,1\le {{j}^{*}}\le n_{B}^{*}, \\ 
      \end{aligned}      
\end{equation}
where $n_{A}^{*}, n_{B}^{*}$ represents the total numbers of features in the feature maps of $A$ and $B$.

Denote the score matrices of original matching and reweighted matching as $\overline{\mathbf{S}}$ and $\overline{\mathbf{S}}^*$, respectively.
According to \cref{thm:rw_attention}, the scores converge in probability:
\begin{equation}
    \begin{split}
     &  \forall i\in [n_A],\forall j\in [n_B],\exists {{i}^{*}}\in [n_{A}^{*}],\exists {{j}^{*}}\in [n_{B}^{*}]: \\ 
     & {{\overline{\mathbf{S}}}_{i,j}}\xrightarrow{P}\overline{\mathbf{S}}_{{{i}^{*}},{{j}^{*}}}^*
    \end{split}
\end{equation}

Because we do not change the dustbin score, the elements in the last rows or the last columns of $\overline{\mathbf{S}}$ and $\overline{\mathbf{S}}^*$ equal to the dustbin score $\alpha$.

Define ${{\mathcal{I}}_{{{n}_{A}}}}({{i}^{*}})=\{k|k\in [{{n}_{A}}],{{k}^{*}}={{i}^{*}}\}$, ${{\mathcal{J}}_{{{n}_{B}}}}({{j}^{*}})=\{l|l\in [{{n}_{B}}],{{l}^{*}}={{j}^{*}}\}$.
We have
\begin{equation}
    \begin{aligned}
        \forall {{i}^{*}}\in [n_{A}^{*}],\forall {{j}^{*}}\in [n_{B}^{*}],\forall i,{i}' & \in {{\mathcal{I}}_{{{n}_{A}}}}({{i}^{*}}),\forall j,{j}'\in {{\mathcal{J}}_{{{n}_{B}}}}({{j}^{*}}): \\ 
       {{\overline{\mathbf{S}}}_{i,j}} & ={{\overline{\mathbf{S}}}_{{i}',{j}'}}, \\ 
       \frac{1}{{{n}_{A}}}\left| {{\mathcal{I}}_{{{n}_{A}}}}({{i}^{*}}) \right| & \xrightarrow{P}{{p}_{A}}({{i}^{*}}), \\ 
       \frac{1}{{{n}_{B}}}\left| {{\mathcal{J}}_{{{n}_{B}}}}({{j}^{*}}) \right| & \xrightarrow{P}{{p}_{B}}({{j}^{*}}). \\ 
      \end{aligned}      
\end{equation}

In original matching, sinkhorn algorithm computes the output assignment as 
\begin{equation} \label{eqn:OT_P}
    \text{OT}_{i,j}={{\mathbf{u}}_{i}}{{\mathbf{K}}_{i,j}}{{\mathbf{v}}_{j}}.
\end{equation}
Each iteration of the algorithm updates vectors $\mathbf{u}, \mathbf{v}$ to $\mathbf{{u}'}, \mathbf{{v}'}$:
\begin{equation}
  {{\mathbf{{u}'}}_{i}}=\frac{{{\mathbf{a}}_{i}}}{{{(\mathbf{Kv})}_{i}}},{{\mathbf{{v}'}}_{j}}=\frac{{{\mathbf{b}}_{j}}}{{{({{\mathbf{K}}^{\top }}\mathbf{{u}'})}_{j}}},
\end{equation}
where ${{\mathbf{K}}_{i,j}}={\mathrm{exp}({\frac{{{\overline{\mathbf{S}}}_{i,j}}}{\varepsilon }}})$.

In reweighted matching, the same algorithm computes the output assignment 
\begin{equation} \label{eqn:OT_Pstar}
    \text{OT}_{{{i}^{*}},{{j}^{*}}}^{{{p}_{A}},{{p}_{B}}}=\mathbf{u}_{{{i}^{*}}}^{*}\mathbf{K}_{{{i}^{*}},{{j}^{*}}}^{*}\mathbf{v}_{{{j}^{*}}}^{*}.
\end{equation}
Each iteration updates vector ${\mathbf{{u}}}^{*}, {\mathbf{{v}}}^{*}$ to ${\mathbf{{u}'}}^{*}, {\mathbf{{v}'}}^{*}$:
\begin{equation}
    {\mathbf{{u}'}}^{*}_{{{i}^{*}}}=\frac{\mathbf{a}_{{{i}^{*}}}^{*}}{{{({{\mathbf{K}}^{*}}{{\mathbf{v}}^{*}})}_{{{i}^{*}}}}}, {\mathbf{{v}'}}^{*}_{{{j}^{*}}}=\frac{\mathbf{b}_{{{j}^{*}}}^{*}}{{{({{\mathbf{K}}^{*}}^{\top }{{{\mathbf{{u}'}}}^{*}})}_{{{j}^{*}}}}},
\end{equation}
where ${{\mathbf{K}}_{i^*,j^*}^*}={\mathrm{exp}({\frac{{{\overline{\mathbf{S}}}_{i^*,j^*}^*}}{\varepsilon }}})$.

At the first iteration, $\mathbf{u}, \mathbf{v}, {\mathbf{{u}}}^{*}, {\mathbf{{v}}}^{*}$ can be initialized with arbitrary positive vectors.
For simplicity, we initialize $\mathbf{u}, \mathbf{v}, {\mathbf{{u}}}^{*}, {\mathbf{{v}}}^{*}$ as $\mathbf{a}, \mathbf{b}, \mathbf{a}^{*}, \mathbf{b}^{*}$, respectively.
Thus, the initialized values satisfy $n_A \mathbf{u}_i = p_A(i^*)^{-1} \mathbf{u}_{i^*}^*$, $n_B \mathbf{v}_j = p_B(j^*)^{-1} \mathbf{v}_{j^*}^*$, $i \in [n_A], j \in [n_B]$, and $\mathbf{u}_{n_A+1} = \mathbf{u}^*_{n_A^*+1} = \mathbf{v}_{n_B+1} = \mathbf{v}^*_{n_B^*+1} = 1$.
Now we employ mathematical induction, and assume that for the $l$-th iteration:
\begin{equation} \label{eqn:lim_uv}
    \begin{split}
      n_A \mathbf{u}_i & \xrightarrow{P} p_A(i^*)^{-1} \mathbf{u}_{i^*}^*,  i \in [n_A], \\
      n_B \mathbf{v}_j & \xrightarrow{P} p_B(j^*)^{-1} \mathbf{v}_{j^*}^*,  j \in [n_B], \\
      \mathbf{u}_{n_A+1} & \xrightarrow{P} \mathbf{u}^*_{n_A^*+1} , \\
      \mathbf{v}_{n_B+1} & \xrightarrow{P} \mathbf{v}^*_{n_B^*+1} .
    \end{split}
\end{equation}
Then for $i \in [n_A], j \in [n_B]$, we have
\begin{equation}
    \begin{aligned}
         {{n}_{A}}{{{\mathbf{{u}'}}}_{i}} & ={{(\mathbf{Kv})}_{i}}^{-1} \\ 
       & ={{\left( \alpha \mathbf{v}_{n_B+1} +\sum\limits_{j=1}^{{{n}_{B}}}{{{\mathbf{K}}_{i,j}}{{\mathbf{v}}_{j}}} \right)}^{-1}} \\ 
       & ={{\left( \alpha \mathbf{v}_{n_B+1} +\sum\limits_{{{j}^{*}}=1}^{n_{B}^{*}}{\frac{1}{{{n}_{B}}}}\sum\limits_{j\in {{\mathcal{J}}_{{{n}_{B}}}}({{j}^{*}})}^{{}}{{{\mathbf{K}}_{i,j}}\cdot ({{n}_{B}}{{\mathbf{v}}_{j}})} \right)}^{-1}} \\ 
       & \xrightarrow{P}{{\left( \alpha \mathbf{v}_{n_B^*+1}^* +\sum\limits_{{{j}^{*}}=1}^{n_{B}^{*}}{{{p}_{B}}({{j}^{*}})}\mathbf{K}_{{{i}^{*}},{{j}^{*}}}^{*}\cdot {{p}_{B}}{{({{j}^{*}})}^{-1}}\mathbf{v}_{{{j}^{*}}}^{*} \right)}^{-1}} \\ 
       & ={{({{\mathbf{K}}^{*}}{{\mathbf{v}}^{*}})}_{{{i}^{*}}}}^{-1} \\ 
       & ={{p}_{A}}{{({{i}^{*}})}^{-1}}\mathbf{{u}'}_{{{i}^{*}}}^{*} \\ 
      \end{aligned}      
\end{equation}
\begin{equation}
    \begin{aligned}
       {{n}_{B}}{{{\mathbf{{v}'}}}_{j}} & ={{({{\mathbf{K}}^{\top }}\mathbf{{u}'})}_{j}}^{-1} \\ 
       & ={{\left( \alpha \mathbf{u'}_{n_A+1} +\sum\limits_{i=1}^{{{n}_{A}}}{{{\mathbf{K}}_{i,j}}{{{\mathbf{{u}'}}}_{i}}} \right)}^{-1}} \\ 
       & ={{\left( \alpha \mathbf{u'}_{n_A+1} +\sum\limits_{{{i}^{*}}=1}^{n_{A}^{*}}{\frac{1}{{{n}_{A}}}}\sum\limits_{i\in {{\mathcal{I}}_{{{n}_{A}}}}({{i}^{*}})}^{{}}{{{\mathbf{K}}_{i,j}}({{n}_{A}}{{{\mathbf{{u}'}}}_{i}})} \right)}^{-1}} \\ 
       & \xrightarrow{P}{{\left( \alpha \mathbf{u'}_{n_A^*+1}^* +\sum\limits_{{{i}^{*}}=1}^{n_{A}^{*}}{{{p}_{A}}({{i}^{*}})}\mathbf{K}_{{{i}^{*}},{{j}^{*}}}^{*}\cdot {{p}_{A}}{{({{i}^{*}})}^{-1}}\mathbf{{u}'}_{{{i}^{*}}}^{*} \right)}^{-1}} \\ 
       & ={{p}_{B}}{{({{j}^{*}})}^{-1}}{{\mathbf{{v}'}}}^{*}_{{{j}^{*}}} \\ 
      \end{aligned}      
\end{equation}
Similarly, we obtain $ \mathbf{u'}_{n_A+1} \xrightarrow{P} \mathbf{u'}^*_{n_A^*+1} $ and $\mathbf{v'}_{n_B+1} \xrightarrow{P} \mathbf{v'}^*_{n_B^*+1} $.

Therefore, \cref{eqn:lim_uv} holds in all iterations and in \cref{eqn:OT_P,eqn:OT_Pstar}. 
For each $i^*, j^*$, by summing up all $\text{OT}_{i,j}$ within ${\mathcal{I}}_{{{n}_{A}}}({{i}^{*}})$ and ${{\mathcal{J}}_{{{n}_{B}}}}({{j}^{*}})$, and examining the convergence value in terms of probability, we prove the validity of \cref{thm:rw_matching} in the context of optimal transport.
The validity about matching using dual-softmax can also be proved using analogous techniques.
\end{proof}
% \begin{equation}
%     \begin{split}
%     &  \mathbf{X}'= \mathbf{X}_Q + \sum_{h = 1}^{H} \mathbf{W}_{VO}^{h} \mathbf{X}_V \mathrm{Sim}(\mathbf{W}_K^{h}\mathbf{X}_K, \mathbf{W}_Q^{h}\mathbf{X}_Q) \\
%     & \mathbf{X}_Q' = \mathbf{X}' + \mathbf{W}_2 \mathrm{Act} (\mathbf{W}_1 X' + \mathbf{b}_1 \mathbf{1}^\top) + \mathbf{b}_2 \mathbf{1}^\top
%     \end{split}
% \end{equation}

% \begin{theorem}\label{thm:rw_attention}
%     Let $F_A$ and $F_B$ be two sequences of i.i.d. feature points, which are sampled from $F_A^*$ and $F_B^*$ according to detection probabilities $p_A$ and $p_B$, respectively.
%     Let $\mathrm{TF}$ and $\mathrm{TF}^{p_A, p_B}$ be any attention network and its reweighted version, as described above. 
%     And the feature points are aggregated as described in \cref{eqn:feature_update}, where $\mathrm{TF_{sparse}} = \mathrm{TF}$ and $\mathrm{TF_{dense}} = \mathrm{TF}^{p_A, p_B}$.
%     For any index pair $i$ and $j$, denote $i$-th and $j$-th output feature points from network $\mathrm{TF}$ as $\mathrm{TF}_{i,j}$, i.e., $\mathrm{TF}_{i,j}(F_A, F_B) = (F'_{A,i}, F'_{B,j})$.
%     As the sizes of $F_A$ and $F_B$ tend to infinity, we have
%     % \begin{equation}
%     %   \lim_{n_A, n_B \to \infty} \mathrm{TF}_{i,j}(F_A, F_B) \stackrel{P}{=} \mathrm{TF}_{i^*,j^*}^{p_A, p_B}(F_A^*, F_B^*).
%     % \end{equation}
%     \begin{equation}
%      \mathrm{TF}_{i,j}(F_A, F_B) \stackrel{P}{\longrightarrow} \mathrm{TF}_{i^*,j^*}^{p_A, p_B}(F_A^*, F_B^*).
%     \end{equation}
% \end{theorem}
}


\bibliographystyle{IEEEtran}
\bibliography{ref/ReweightedMatcher.bib}

% \newpage

 

% \end{IEEEbiography}



\begin{IEEEbiography}[{\includegraphics[width=1in,height=1.25in,clip,keepaspectratio]{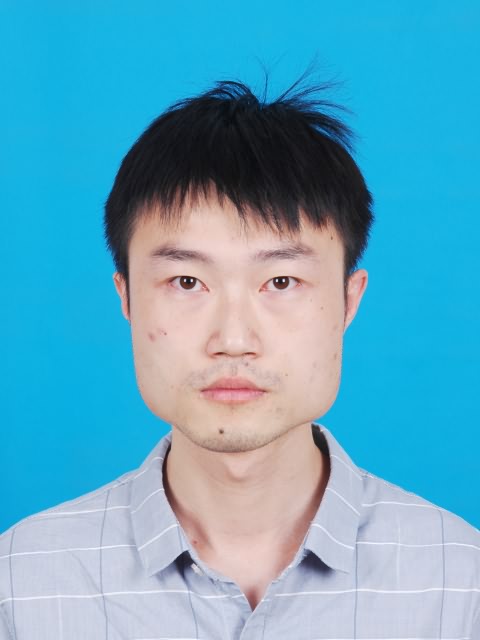}}]{Ya Fan}
  received his M.Sc. degree in the School of Electronic Information Engineering, Beihang University, Beijing 100191, China. He is currently pursuing the Ph.D. degree in Electronic Information Engineering from Beihang University, focusing on computer vision for simultaneous localization and mapping. His research interests include deep learning, computer vision and robotic science.
\end{IEEEbiography}

\begin{IEEEbiography}[{\includegraphics[width=1in,height=1.25in,clip,keepaspectratio]{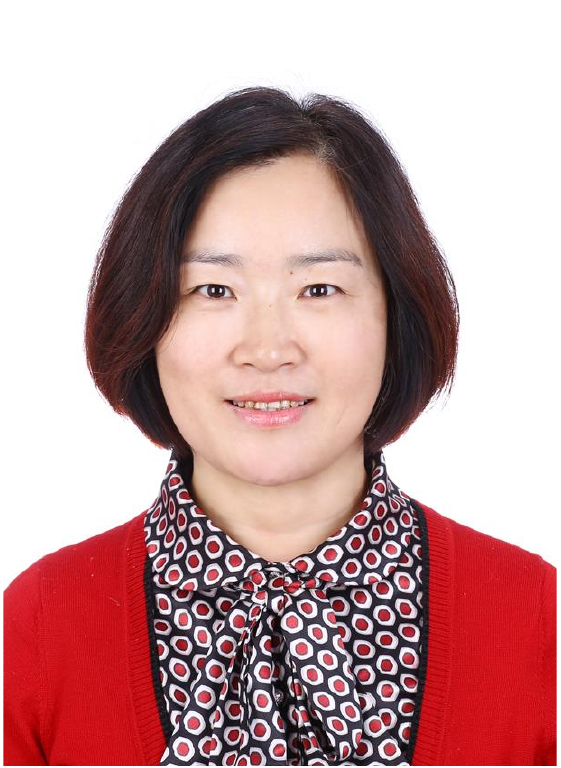}}]{Rongling Lang}
  was born in 1975. She received the Ph.D. degree in automatic control from Northwestern Polytechnical University, Xi'an, China, in 2005. She is currently with the School of Electronic and Information Engineering, Beihang University, Beijing, China. Her research interests include machine learning and vision navigation.
\end{IEEEbiography}

% \vfill

\end{document}